\documentclass[12pt]{article}
\RequirePackage[OT1]{fontenc}
\RequirePackage{amsthm,amsmath,amssymb,subcaption,graphicx,epstopdf,enumerate,lmodern}
\RequirePackage[round,colon,authoryear]{natbib}
\RequirePackage[colorlinks,citecolor=blue,urlcolor=blue]{hyperref}

\usepackage{amsmath}
\usepackage{amssymb}
\usepackage{amscd}
\usepackage{bbm}

\def\eps{\varepsilon}
\font\tencmmib=cmmib10 \skewchar\tencmmib '60
\newfam\cmmibfam
\textfont\cmmibfam=\tencmmib

\def\lessim{\ \lower4pt\hbox{$
\buildrel{\displaystyle <}\over\sim$}\ }
\def\gessim{\ \lower4pt\hbox{$\buildrel{\displaystyle >}
\over\sim$}\ }

\usepackage{algorithm,algorithmicx,algpseudocode}

\newtheorem{theorem}{Theorem}
\newtheorem{lemma}{Lemma}

\renewcommand{\hat}{\widehat}

\renewcommand{\hat}{\widehat}

\newcommand{\bfm}[1]{\ensuremath{\mathbf{#1}}}

   \def\bA{\bfm A}  
   \def\bB{\bfm B}

   \def\bD{\bfm D}  
\def\be{\bfm e}   \def\bE{\bfm E}  \def\EE{\mathbb{E}}

   \def\bM{\bfm M}

   \def\bP{\bfm P}  \def\PP{\mathbb{P}}
     
     \def\RR{\mathbb{R}}
\def\bs{\bfm s}   \def\bS{\bfm S}  
     
\def\bu{\bfm u}   \def\bU{\bfm U}

   \def\bX{\bfm X}  
   \def\bY{\bfm Y}  
   \def\bZ{\bfm Z}  

\def\calA{{\cal  A}}

\def\calE{{\cal  E}}

\def\calM{{\cal  M}}

\def\calP{{\cal  P}}

\def\calT{{\cal  T}}

\newcommand{\bfsym}[1]{\ensuremath{\boldsymbol{#1}}}

            \def\bDelta {\bfsym {\Delta}}

 \def\beps{\bfsym \varepsilon}          \def\bepsilon{\bfsym \varepsilon}

\DeclareMathOperator{\sgn}{sgn}

\DeclareMathOperator{\Var}{Var}

\def\eps{\varepsilon}
\def\beps{\mbox{\boldmath$\eps$}}
\def\newpage{\vfill\eject}

\usepackage{tikz}
\usetikzlibrary{decorations.pathreplacing,calc}
\newcommand{\tikzmark}[1]{\tikz[overlay,remember picture] \node (#1) {};}

\newcommand*{\AddNote}[4]{%
    \begin{tikzpicture}[overlay, remember picture]
        \draw [decoration={brace,amplitude=0.5em},decorate,ultra thick,black]
            ($(#3)!(#1.north)!($(#3)-(0,1)$)$) --  
            ($(#3)!(#2.south)!($(#3)-(0,1)$)$)
                node [align=center, text width=2.5cm, pos=0.5, anchor=west] {#4};
    \end{tikzpicture}
}%

\newdimen\biblioindent    \biblioindent=30pt

 at 8truept

\def\sgn{\mbox{sgn}}

\def\eps{\varepsilon}

\newcommand{\beq}{\begin{equation}}
  \newcommand{\eeq}{\end{equation}}
\newcommand{\beqn}{\begin{eqnarray}}
  \newcommand{\eeqn}{\end{eqnarray}}
\newcommand{\beqnn}{\begin{eqnarray*}}
  \newcommand{\eeqnn}{\end{eqnarray*}}

\def\nnz{{\rm nnz}}
\def\sr{{\rm sr}}

\newcounter{CondCounter}

\usepackage{mathtools}
\DeclarePairedDelimiter\ceil{\lceil}{\rceil}

\begin{document}

\title{Effective Tensor Sketching via Sparsification$^\ast$}
\author{Dong Xia and Ming Yuan$^\dag$\\
Columbia University}
\date{(\today)}

\maketitle

\footnotetext[1]{
This research was supported by NSF Grant DMS-1721584, and NIH Grant 1U54AI117924-01.}
\footnotetext[2]{
Address for Correspondence: Department of Statistics, Columbia University, 1255 Amsterdam Avenue, New York, NY 10027.}

\begin{abstract}
In this paper, we investigate effective sketching schemes via sparsification for high dimensional multilinear arrays or tensors. More specifically, we propose a novel tensor sparsification algorithm that retains a subset of the entries of a tensor in a judicious way, and prove that it can attain a given level of approximation accuracy in terms of tensor spectral norm with a much smaller sample complexity when compared with existing approaches. In particular, we show that for a $k$th order $d\times\cdots\times d$ cubic tensor of {\it stable rank} $r_s$, the sample size requirement for achieving a relative error $\varepsilon$ is, up to a logarithmic factor, of the order $r_s^{1/2} d^{k/2} /\varepsilon$ when $\eps$ is relatively large, and $r_s d /\varepsilon^2$ and essentially optimal when $\varepsilon$ is sufficiently small. It is especially noteworthy that the sample size requirement for achieving a high accuracy is of an order independent of $k$. To further demonstrate the utility of our techniques, we also study how higher order singular value decomposition (HOSVD) of large tensors can be efficiently approximated via sparsification. 
\end{abstract}

\newpage

\section{Introduction}
Massive datasets are being generated everyday across diverse fields and can often be formatted into matrices or higher order tensors. For example, in biomedical research, huge data matrices and tensors arise in gene expression analysis \citep[see, e.g.,][]{kluger2003spectral}, protein-to-protein interaction \citep[see, e.g.,][]{stelzl2005human}, and MRI image analysis \citep[see, e.g.,][]{smith2004advances}. They also occur frequently in statistical physics \citep[see, e.g.,][]{orus2014practical,cichocki2015tensor}, video processing \citep[see, e.g.,][]{li2010tensor, liu2013tensor}, and analyzing large graphs and social networks \citep[see, e.g.,][]{clauset2004finding, abadi2016tensorflow, scott2017social}, to name a few. As the size of these data matrices or tensors grows, it becomes costly and sometimes prohibitively expensive to store, communicate or manipulate them. This naturally brings about the task of ``sketching'': approximate the original data matrices or tensors with a more manageable amount of sketches.

In the case of data matrices, numerous sketching approaches have been proposed in recent years. See \cite{woodruff2014sketching} for a recent review. A popular idea behind many of these approaches is {\it sparsification} -- creating a sparse matrix by zeroing out some entries of the original data matrix. Sparse sketching of a large data matrix not only reduces space complexity but also allows for efficient computations. See, e.g., \cite{frieze2004fast, arora2006fast, achlioptas2007fast, drineas2011note, achlioptas2013near, krishnamurthy2013low}, among others. The main purpose of this article is to investigate to what extent sparsification can be used to effectively sketch higher order tensors. There have been some recent attempts along this direction. In particular, our work is inspired by \cite{nguyen2015tensor} who showed that for a $k$th order cubic tensor $\bA\in \RR^{d\times\cdots\times d}$, there is a randomized sparsification scheme that yields another tensor $\tilde{\bA}$ of same dimension but with
\begin{equation}
\label{eq:nsample}
\nnz(\tilde{\bA})=\tilde{O}_p\left(d^{k/2}\sr(\bA)\over \varepsilon^2\right),\qquad {\rm as\ } d\to\infty,
\end{equation}
such that
$$
\|\tilde{\bA}-\bA\|\le \varepsilon \|\bA\|.
$$
Here, $\nnz(\cdot)$ stands for the number of nonzero entries of a tensor, $\sr(\bA)=\|\bA\|_{\rm F}^2/\|\bA\|^2$ is the so-called stable rank \citep[see, e.g.,][]{achlioptas2013near, nguyen2015tensor} of a tensor $\bA$, $\|\cdot\|$ is the usual tensor spectral norm, and $\tilde{O}(\cdot)$ means $O(\cdot)$, up to a certain polynomial of logarithmic factor. Similar results have also been obtained by \cite{bhojanapalli2015new} in the case when $k=3$. On the one hand, the sample size requirement given by \eqref{eq:nsample} is satisfying because it is essentially optimal in the matrix case, that is $k=2$. See, e.g., \cite{achlioptas2013near}. On the other hand, the exponential dependence on $k$ suggests a large amount of entries still need to be retained to yield a good approximation. Our goal is to investigate if this aspect could be improved.

In particular, we propose a novel tensor sparsification algorithm that randomly retain entries from $\bA$ in a judicious way to yield a tensor $\hat{\bA}^{\rm SPA}$ such that
$$
\|\hat{\bA}^{\rm SPA}-\bA\|\le \varepsilon\|\bA\|,
$$
and
\begin{equation}
\label{eq:sample}
\nnz(\hat{\bA}^{\rm SPA})=\tilde{O}_p\left(\max\left\{{d\cdot \sr(\bA)\over \varepsilon^2},{d^{k/2}\cdot\sr(\bA)^{1/2}\over \varepsilon}\right\}\right).
\end{equation}
Here, to fix ideas, we focus on the case of cubic tensors although our results deal with more general rectangular tensors as well. This sample size requirement significantly improves those earlier ones. Especially if a high accuracy approximation is sought, that is $\varepsilon\le \sr(\bA)\cdot d^{-k/2+1}$, then our sparsification algorithm can achieve relative approximation error $\varepsilon$ in terms of tensor spectral norm by retaining as few as $\tilde{O}_p(d\cdot\sr(\bA)\cdot\varepsilon^{-2})$ entries of $\bA$, regardless of the order of the tensor. Furthermore, for larger $\varepsilon$, the number of nonzero entries we keep is smaller than $\tilde{\bA}$ by a factor of $\sr(\bA)^{1/2}\varepsilon^{-1}$.

Similar to many other sparsification algorithms, we treat different entries according to their magnitude: large entries are always kept, and moderate ones are sampled proportion to their square values. The key difference between our approach and the existing ones is in the treatment of small entries. Instead of zeroing them out as, for example, \cite{nguyen2015tensor}, we sample them in a uniform fashion, which proves to be essential for obtaining good approximation with tighter number of nonzero entries. This modification is motivated by the concentration behavior of randomly sampled tensors recently observed by \cite{yuan2016tensor, yuan2017incoherent, xia2017polynomial}.

To demonstrate the effectiveness of our tensor sketching schemes, we show how they can be used for efficient approximation of the leading singular spaces from higher order singular value decomposition (HOSVD). Let $\bU_j\in \RR^{d\times r}$ be the top $r$ left singular vectors of the flattening of $\bA$ along its $j$th mode. We show that it is possible to construct an approximation $\hat{\bU}_j$ obeying
$$
\|\hat{\bU}_j\hat{\bU}_j^\top-\bU_j\bU_j^\top\|\le \varepsilon,
$$
if we retain
$$
\tilde{O}_p\left(\max\left\{{rd\over\varepsilon^2},{rd^{k/2}\over \varepsilon}\right\}\right)
$$
carefully chosen entries
As before, we note that for high accuracy approximations, the sample complexity is essentially independent of the order of the tensor.
Although our primary focus is on higher order tensors, as a byproduct, our results indicate that our sparsification scheme improves the sample complexity of earlier approaches for approximating the singular vectors of highly rectangular matrix.

The rest of the paper is organized as follows. We first discuss the new tensor sparsification algorithm in Section~\ref{sec:SPA}
. In Section \ref{sec:hosvd} we consider the application to HOSVD.
All proofs are relegated to Section~\ref{sec:proof}.

\section{Tensor Sparsification}
\label{sec:SPA}

Sketches of a tensor $\bA\in \RR^{d_1\times \ldots \times d_k}$ are its approximations. We consider measuring the quality in terms of relative error in terms of tensor spectral norm. Recall that the spectral norm of a tensor $\bB\in \RR^{d_1\times \ldots \times d_k}$ is defined as
$$
\|\bB\|=\underset{\bu_j\in\RR^{d_j}, \|\bu_j\|_{\ell_2}\leq 1}{\sup}\ \left\langle\bB, \bu_1\otimes\ldots\otimes\bu_k\right\rangle.
$$
We seek an approximation $\hat{\bA}$ of $\bA$ such that
$$
\|\hat{\bA}-\bA\|\le \varepsilon \|\bA\|,
$$
for some $\varepsilon\in (0,1)$.

We first consider sketching a tensor by sparsification. The idea is to systematically zero out entries of $\bA$ and scale the remaining entries to yield a good approximation of $\bA$. We focus here on sparsification strategies that are carried out in an entry-by-entry fashion. Our approach can be characterized as {\it keeping} large entries, {\it sampling proportionally} moderate entries, and {\it sampling uniformly} small entries. The key is determining how to classify entries into these categories, and how to sample the moderate entries, so that the number of nonzero entries retained are as small as possible. Details are presented in Algorithm \ref{algo:sparsify}.

\begin{algorithm}
 \caption{Tensor Sparsification}\label{algo:sparsify}
  \begin{algorithmic}[2]
  \State Input: $\bA\in\RR^{d_1\times\ldots\times d_k}$, sampling budget $1\le n\le d_1\cdots d_k$.
  \State Output: $\hat{\bA}^{\rm SPA}\in\RR^{d_1\times\ldots\times d_k}$.
  \For{$i_1\in [d_1], i_2\in[d_2],\ldots, i_k\in[d_k]$}
	\If{$|A(i_1,\ldots,i_k)|\geq \|\bA\|_{\rm F}/n^{1/2}$,}\tikzmark{top2}
	\State $\hat{A}(i_1,\ldots,i_k)=A(i_1,\ldots,i_k)$.
	\EndIf   \tikzmark{bottom2}
   
	\If{$|A(i_1,\ldots,i_k)|/\|\bA\|_{\rm F}\in\Big(\frac{1}{(d_1\cdots d_k)^{1/2}}, \frac{1}{n^{1/2}}\Big)$,}\tikzmark{top3}
	\begin{equation*}
	\hat{A}(i_1,\ldots,i_k)=
	\begin{cases}
	\frac{A(i_1,\ldots,i_k)}{P(i_1,\ldots,i_k)},& \textrm{with probability\ } P(i_1,\ldots,i_k):=\frac{n A^2(i_1,\ldots,i_k)}{\|\bA\|_{\rm F}^2}\\
	0,& \textrm{with probability\ } 1-P(i_1,\ldots,i_k).
	\end{cases}
	\end{equation*}
	\EndIf   \tikzmark{bottom3}
   
         \If{$|A(i_1,\ldots,i_k)|\leq {\|\bA\|_{\rm F}}/{(d_1\ldots d_k)^{1/2}}$,}\tikzmark{top1}
         \State
         \begin{equation*}
          \hat{A}(i_1,\ldots, i_k)=
          \begin{cases}
          \frac{A(i_1,\ldots,i_k)}{P(i_1,\ldots,i_k)},& \textrm{ with probability\ } P(i_1,\ldots,i_k):=\frac{n}{d_1d_2\cdots d_k}\tikzmark{right}\\
          0,& \textrm{ with probability\ } 1-P(i_1,\ldots,i_k)
          \end{cases}
          \end{equation*}
         \EndIf\tikzmark{bottom1}
   
   \EndFor
   \State Output: $\hat\bA^{\rm SPA}=\hat\bA$.
  \end{algorithmic}
  \AddNote{top1}{bottom1}{right}{Small\\ Entries}
  \AddNote{top2}{bottom2}{right}{Large\\ Entries}
  \AddNote{top3}{bottom3}{right}{Moderate\\ Entries}
\end{algorithm}

In particular, we keep all entries whose absolute value is greater than $n^{-1/2}\|\bA\|_{\rm F}$, sample uniformly all entries whose absolute value is smaller than $(d_1\cdots d_k)^{-1/2}\|\bA\|_{\rm F}$, and sample proportional to their squared values entries whose absolute value is in-between. Here $n$ is a sampling parameter. Note that $\EE[\nnz(\hat{\bA}^{\rm SPA})]\le 2n$. And it is not hard to see, by Chernoff bound, that $\nnz(\hat{\bA}^{\rm SPA})=O_p(n)$. In other words, $n$ represents essentially the targeted sampling budget.

We note that our sparsification algorithm is similar to the one proposed earlier by \cite{nguyen2015tensor}. But the two schemes also have several key differences. The main difference between the two algorithms is their treatment of ``small'' entries. \cite{nguyen2015tensor} suggests to zero them out, while ours sample them in a uniform fashion. This is largely motivated by the concentration behavior of randomly sampled tensors observed earlier. In particular, it can be shown that a uniformly sampled tensor concentrates much sharply around its mean if its entries are sufficiently small \citep[see, e.g.,][]{yuan2016tensor}. Therefore, instead of discarding small entries, we could derive a good estimate of them by sampling uniformly. Another subtle difference between the two algorithm is in the criteria for ``small'' entries. Our criterion for ``small'' entries is that their absolute values are smaller than $(d_1\cdots d_k)^{-1/2}\|\bA\|_{\rm F}$, whereas \cite{nguyen2015tensor} treats only cubic tenors, that is $d_1=d_2=\cdots=d_k=:d$, and small entries of their scheme are those smaller than $n^{-1/2}d^{-k/4}\|\bA\|_{\rm F}\log^{k/2}d$. 

We now present the performance bounds for our sparsification algorithm.
\begin{theorem}\label{thm:sparsify}
Let $\bA\in\RR^{d_1\times\ldots\times d_k}$ and $\hat{\bA}^{\rm SPA}$ be the output from Algorithm~\ref{algo:sparsify} with sampling budget $n$. There exists a constant $C>0$ depending on $k$ only such that if for any $\alpha\geq 4\log(k\log d_{\max})$ and $\varepsilon\in(0,1)$, if
$$
n\geq C\max\left\{\alpha^4\frac{d_{\max}\cdot\sr(\bA)}{\eps^2}\log^{2}d_{\max},\ \alpha^2\frac{(d_1\cdots d_k\cdot\sr(\bA))^{1/2}}{\eps}\log^{k+4}d_{\max}\right\},
$$
then, with probability at least $1-d_{\max}^{-\alpha}$,
$$
\|\hat{\bA}^{\rm SPA}-\bA\|\leq \eps\|\bA\|,
$$
where $d_{\max}=\max\{d_1,\ldots,d_k\}$.
\end{theorem}

In the light of Theorem \ref{thm:sparsify}, we can achieve relative error $\eps$ in terms of tensor spectral norm with a sparse tensor such that
\begin{eqnarray*}
\nnz(\hat{\bA}^{\rm SPA})=
\begin{cases}
\tilde{O}\bigg(\eps^{-2}\cdot{d_{\max}\cdot\sr(\bA)}\bigg),& \textrm{ if } \varepsilon\leq {d_{\max}\cdot\sr(\bA)^{1/2}}\cdot{(d_1\ldots d_k)^{-1/2}};\\
\tilde{O}\bigg(\eps^{-1}\cdot(d_1\ldots d_k\cdot\sr(\bA))^{1/2}\bigg),& \textrm{ otherwise. }
\end{cases}
\end{eqnarray*}
This significant improves earlier work by \cite{bhojanapalli2015new} and \cite{nguyen2015tensor}. It is worth noting that for small $\eps$, or high accuracy approximation, the number of nonzero entries of $\hat{\bA}^{\rm SPA}$ is of the order $\eps^{-2}d_{\max}\cdot\sr(\bA)$, regardless of $k$. This, in particular, is known to be optimal in the matrix ($k=2$) case \citep[see, e.g.,][]{achlioptas2013near}.

The main technical tool for proving Theorem~\ref{thm:sparsify} is the following concentration inequality for random tensors which might be of independent interest.
\begin{theorem}\label{thm:conineq}
Let $\bA\in\RR^{d_1\times \ldots\times d_k}$ and $\bP\in [0,1]^{d_1\times\ldots\times d_k}$ be two fixed tensors, $\bDelta\in\{0,1\}^{d_1\times\ldots\times d_k}$ be a random tensor such that $\EE\Delta(i_1,\ldots,i_k)=P(i_1,\ldots,i_k)$. Define a random tensor $\hat{\bA}\in\RR^{d_1\times\ldots\times d_k}$ by
$$\hat{A}(i_1,\ldots,i_k)=A(i_1,\ldots,i_k)\Delta(i_1,\ldots,i_k)/P(i_1,\ldots,i_k).$$ 
Then, there exist absolute constants $C_1,C_2,C_3>0$ such that for any $\alpha>0$, with probability at least $1-3d_{\max}^{-\alpha}$,
\begin{eqnarray*}
\|\hat{\bA}-\bA\|\leq C_1\Big(\big(\sum_{j=1}^kd_j\big)^{1/2}+\alpha k\log d_{\max}\Big)\alpha_{2,\infty}(\bA,\bP)+C_2\alpha k^3\log^{k+2}(d_{\max})\sqrt{\nu}\alpha_\infty(\bA,\bP),
\end{eqnarray*}
where 
$$\nu=C_3\alpha\max\big\{\beta(\bP), k\log d_{\max}\big\},$$
$$
\beta(\bP)=\max_{j=1,\ldots,k}\ \max_{i_1,\ldots,i_{j-1},i_{j+1},\ldots,i_k} \sum_{i_{j}=1}^{d_{j}}P(i_1,\ldots,i_k),
$$
$$
\alpha_{\infty}(\bA,\bP)=\max_{i_j\in[d_j], j=1,\ldots,k}\frac{|A(i_1,\ldots,i_k)|}{P(i_1,\ldots,i_k)},
$$
and
$$
\alpha_{2,\infty}(\bA,\bP)=\max_{i_j\in[d_j], j=1,\ldots,k}\bigg(\frac{A^2(i_1,\ldots,i_k)}{P(i_1,\ldots,i_k)}\bigg)^{1/2}.
$$
Here we follow the convention that $0/0=0$.
\end{theorem}

\section{HOSVD via Tensor Sketching}
\label{sec:hosvd}

To further illustrate the merits of the sketching schemes introduced earlier, we now consider a specific application to HOSVD, a popular technique for analyzing high dimensional tensor data. See, e.g., \cite{kolda2009tensor, sidiropoulos2017tensor} and references therein. 

For a $k$-th order tensor $\bA\in\RR^{d_1\times\ldots\times d_k}$, let $\bM_j=\calM_j(\bA)\in\RR^{d_j\times d_{-j}}$ be its $j$-th matricization where $1\leq j\leq k$, that is,
$$
\calM_j(\bA)\Big(i_j, \sum_{s=1,s\neq j}^k(i_s-1)\Big(\prod_{s'=s+1, s'\neq j}^k d_{s'}\Big)+1\Big)=A(i_1,\ldots,i_k),\quad \forall i_j\in[d_j], 1\leq j\leq k.
$$
Here $d_{-j}=(d_1\cdots d_k)/d_j$. Denote by $\bU_j^{(r_j)}$ the collection of the top $r_j$ left singular vectors of $\bM_j$. Clearly, $\bU_j^{(r_j)}$ is computable via the standard matrix singular value decomposition on $\bM_j$ whose computation complexity is $O(d_jd_1d_2\ldots d_k)$, see \cite{golub2012matrix}. Efficient computation of singular value decomposition for large matrices is an actively researched topic in numerical algebra and computational science. See \cite{berry1992large, kobayashi2001estimation, achlioptas2007fast, holmes2007fast, drineas2011note, menon2011fast}, among numerous others.

A general idea is to first obtain an approximation of $\bM_j$, say $\hat{\bM}_j\in\RR^{d_j\times {d_{-j}}}$, that is amenable for fast computation of singular value decomposition; and then approximate $\bU_j^{(r_j)}$ by the top left singular vectors of $\hat{\bM}_j$. In particular, sparsification is commonly used to yield $\hat{\bM}_j$. Denote by $\bDelta_j=\hat\bM_j-\bM_j$ and by $\hat\bU_j^{(r_j)}$ the leading $r_j$ left singular vectors of $\hat\bM_j$. By Davis-Kahan Theorem \citep{davis1970rotation}, we get
\begin{equation}\label{eq:Davis-Kahan}
\big\|\hat{\bU}_j^{(r_j)}\big(\hat{\bU}_j^{(r_j)}\big)^{\top}-\bU_j^{(r_j)}\big(\bU_j^{(r_j)}\big)^{\top}\big\|\leq \frac{2\|\bDelta_j\|}{\bar{g}_{r_j}(\bM_j)}
\end{equation}
where $\sigma_k(\cdot)$ denotes the $k$-th singular value, and
$$
\bar{g}_{r_j}(\bM_j)=\sigma_{r_j}(\bM_j)-\sigma_{r_j+1}(\bM_j),
$$
is the $r_j$-th eigengap. In particular, we can consider applying this strategy by taking $\hat{\bM}_j=\calM_j(\hat{\bA}^{\rm SPA})$. The following result characterizes its performance.

\begin{theorem}\label{thm:SVD}
Let $\bU_j^{(r_j)}$ and $\hat\bU_j^{(r_j)}$ be the top $r_j$ left singular vectors of $\calM_j(\bA)$ and $\calM_j(\hat{\bA}^{\rm SPA})$ respectively. Then there exists a constant $C>0$ depending on $k$ only such that for any $t>0$,
\begin{eqnarray*}
\big\|\hat{\bU}_j^{(r_j)}\big(\hat{\bU}_j^{(r_j)}\big)^{\top}-\bU_j^{(r_j)}\big(\bU_j^{(r_j)}\big)^{\top}\big\|\\
\leq C\frac{\|\bM_j\|_{\rm F}}{\bar{g}_{r_j}(\bM_j)}\bigg(\sqrt{\frac{d_1\ldots d_k(t+\log d_{\max})}{nd_j}}+\frac{(d_1\ldots d_k)^{1/2}(t+\log d_{\max})}{n}\bigg),
\end{eqnarray*}
with probability at least $1-e^{-t}$.
\end{theorem}
By Theorem~\ref{thm:SVD}, in the case when $\frac{\|\bM_j\|_{\rm F}}{\bar{g}_{r_j}(\bM_j)}=O(\sqrt{r_j})$, we can ensure
$$\big\|\hat{\bU}_j^{(r_j)}\big(\hat{\bU}_j^{(r_j)}\big)^{\top}-\bU_j^{(r_j)}\big(\bU_j^{(r_j)}\big)^{\top}\big\|\leq \varepsilon$$
by taking
\begin{equation}\label{eq:nbound1}
n\geq C\cdot\max\bigg\{\frac{r_jd_1\ldots d_k}{d_j\varepsilon^2}, \frac{(r_jd_1\ldots d_k)^{1/2}}{\varepsilon}\bigg\}\log d_{\max}.
\end{equation}
A critical fact that is neglected by this approach is that we are interested in approximating the left singular vectors of a potentially very ``fat'' matrix because $d_{-j}$ is generally much larger than $d_j$. As such, this type of approach turns out to be suboptimal for our purpose. 

Alternatively, we adopt a new spectral method similar in spirit to a recent proposal from \cite{xia2017polynomial}. More specifically, we shall approximate $\bU_j^{(r_j)}$ by the leading eigenvectors of an approximation of $\bM_j\bM_j^\top$ instead. In particular, we can run Algorithm \ref{algo:sparsify} twice to obtain two {\it independent} sparsifications of $\bA$, denoted by $\hat{\bA}^{\rm SPA}_1$ and $\hat{\bA}^{\rm SPA}_2$, and then proceed to approximate $\bM_j\bM_j^\top$ by $\calM_j(\hat{\bA}^{\rm SPA}_1)\calM_j(\hat{\bA}^{\rm SPA}_2)^\top$. Details are presented in Algorithm~\ref{algo:HOSVD_SPA}.
\begin{algorithm}
 \caption{Computing HOSVD via Tensor Sparsification}\label{algo:HOSVD_SPA}
  \begin{algorithmic}[2]
  \State Input: $\bA\in\RR^{d_1\times\ldots\times d_k}$, sampling budget $n\geq 1$.
  \State Output: the $r_j$ leading left singular vectors $\hat\bU_j^{(r_j)}$ as an estimate of HOSVD of $\calM_j(\bA)$.
  \State Run Algorithm \ref{algo:sparsify} on $\bA$ with sampling budget $n$. Denote the output by $\hat{\bA}^{\rm SPA}_1$.
  \State Run Algorithm \ref{algo:sparsify} on $\bA$ with sampling budget $n$. Denote the output by $\hat{\bA}^{\rm SPA}_2$.
  \State Compute $\hat\bU_j^{(r_j)}$ as the $r_j$ leading left singular vectors of $\calM_j(\hat{\bA}^{\rm SPA}_1)\calM_j(\hat{\bA}^{\rm SPA}_2)^\top$.
  \State Output $\hat\bU_j^{(r_j)}$.
  \end{algorithmic}
\end{algorithm}

The following theorem provides the performance bound for approximate the singular space $\bU_j^{(r_j)}$s.
\begin{theorem}\label{thm:SVD2}
Denote by $\bU_j^{(r_j)}$ the $r_j$ leading left singular vectors of $\calM_j(\bA)$. Let $\hat\bU_j^{(r_j)}$ be the output from Algorithm~\ref{algo:HOSVD_SPA}. There exists a constant $C>0$ such that for any $\alpha\geq 1$ and $\varepsilon\in(0,1)$, if
$$
n\geq C\alpha\bigg(\frac{d_j\log d_{\max}}{\varepsilon^2}\frac{\|\bA\|_{\rm F}^2\sigma_{\max}^2(\bM_j)}{\bar{g}_{r_j}^2(\bM_j\bM_j^\top)}+\frac{(d_1\ldots d_k)^{1/2}\log d_{\max}}{\varepsilon}\frac{\|\bA\|_{\rm F}^2}{\bar{g}_{r_j}(\bM_j\bM_j^\top)}\bigg),
$$
then
$$\big\|\hat\bU_j^{(r_j)}\big(\hat\bU_j^{(r_j)}\big)^{\top}-\bU_j^{(r_j)}\big(\bU_j^{(r_j)}\big)^{\top}\big\|\leq \varepsilon,$$
with probability at least $1-d_{\max}^{-\alpha}$.
\end{theorem}

From Theorem~\ref{thm:SVD2}, if $\frac{\|\bA\|_{\rm F}^2}{\bar{g}_{r_j}(\bM_j\bM_j^\top)}=O(r_j)$ and $\frac{\|\bA\|_{\rm F}^2\sigma_{\max}^2(\bM_j)}{\bar{g}_{r_j}^2(\bM_j\bM_j^\top)}=O(r_j)$, then the required sample complexity for sparsification is
$$
\tilde{O}_p\bigg(\frac{r_jd_j\log d_{\max}}{\varepsilon^2}+\frac{r_j(d_1\ldots d_k)^{1/2}\log d_{\max}}{\varepsilon}\bigg).
$$
It is worth noting that, even though our main focus is on higher order tensors, in the case of matrices ($k=2$) this sample complexity compares favorable with other sparsification techniques that have been developed for computing singular vectors. For example, consider computing the top $r$ left singular vectors of a $d_1\times d_2$ ($d_1\le d_2$) matrix. The approach from \cite{achlioptas2007fast} needs to sample
$$\tilde{O}_p\Big(\frac{rd_1d_2^2}{\varepsilon^2} \cdot \frac{\max_{i,j}|A(i,j)|^2}{\|\bA\|_{\rm F}^2}\Big)$$
entries; the technique of \cite{drineas2006subspace} requires
$$\tilde{O}_p\Big(\frac{rd_2}{\varepsilon^2}\Big)$$
entries. These are to be compared with Algorithm \ref{algo:HOSVD_SPA} which needs
$$\tilde{O}_p\Big(\frac{rd_1}{\varepsilon^2}+\frac{r(d_1d_2)^{1/2}}{\varepsilon}\Big)$$
sampled entries, which could be much smaller than the previous two when $d_1\ll d_2$.

\section{Proofs}\label{sec:proof}

We now present the proofs to our main results.
\subsection{Proof of Theorem~\ref{thm:sparsify}}
Theorem~\ref{thm:sparsify} follows immediately from the concentration bound for $\|\hat{\bA}^{\rm SPA}-\bA\|$ below.
\begin{lemma}\label{lemma:sparsify}
Let $\bA\in\RR^{d_1\times\ldots\times d_k}$ and $\hat{\bA}^{\rm SPA}$ be the output from Algorithm~\ref{algo:sparsify} with sampling budget $n$. Then there exist absolute constants $C_1, C_2>0$ such that, for any $\alpha\geq 4\log(k\log d_{\max})$, the following bound holds with probability at least $1-d_{\max}^{-\alpha}$:
\begin{eqnarray*}
\|\hat\bA^{\rm SPA}-\bA\|\leq C_1\alpha^2k^4\log(d_{\max})\sqrt{\frac{\|\bA\|_{\rm F}^2d_{\max}}{n}}+C_2 \alpha^2k^5\log^{k+4}(d_{\max})\frac{(d_1\ldots d_k)^{1/2}\|\bA\|_{\rm F}}{n}.
\end{eqnarray*}
\end{lemma}

\begin{proof}[Proof of Lemma~\ref{lemma:sparsify}]
Given $\bA$,  we define the disjoint subsets of $[d_1]\times\ldots\times [d_k]$
$$
\Omega_1=\big\{(i_1,\ldots,i_k): |A(i_1,\ldots,i_k)|\leq \|\bA\|_{\rm F}/(d_1\ldots d_k)^{1/2}\big\},
$$
$$
\Omega_2=\bigg\{(i_1,\ldots,i_k): |A(i_1,\ldots,i_k)|/\|\bA\|_{\rm F}\in\Big(\frac{1}{(d_1\ldots d_k)^{1/2}}, \frac{1}{n^{1/2}}\Big)\bigg\},
$$
and
$$
\Omega_3=\big\{(i_1,\ldots,i_k): |A(i_1,\ldots,i_k)|\geq \|\bA\|_{\rm F}/n^{1/2}\big\}.
$$
Note that $\Omega_1,\Omega_2,\Omega_3$ are non-random subsets for given $\bA$.
Then, 
\begin{eqnarray*}
\|\hat{\bA}^{\rm SPA}-\bA\|\leq \|\hat{\bA}^{\rm SPA}_{\Omega_1}-\bA_{\Omega_1}\|+\|\hat{\bA}^{\rm SPA}_{\Omega_2}-\bA_{\Omega_2}\|+\|\hat{\bA}^{\rm SPA}_{\Omega_3}-\bA_{\Omega_3}\|.
\end{eqnarray*}
By definition of $\hat{\bA}^{\rm SPA}$ in Algorithm~\ref{algo:sparsify}, we have $\|\hat{\bA}^{\rm SPA}_{\Omega_3}-\bA_{\Omega_3}\|=0$ so that it suffices to bound $\|\hat{\bA}^{\rm SPA}_{\Omega_1}-\bA_{\Omega_1}\|$ and $\|\hat{\bA}^{\rm SPA}_{\Omega_2}-\bA_{\Omega_2}\|$.

\paragraph*{Upper bound of $\|\hat{\bA}^{\rm SPA}_{\Omega_1}-\bA_{\Omega_1}\|$.} In order to apply Theorem~\ref{thm:conineq}, we introduce auxiliary tensors $\bB$ and $\widetilde\bB$ such that $\bB_{\Omega_1}=\bA_{\Omega_1}$ and $\bB_{\Omega_1^{\dagger}}={\bf 0}$, where $\Omega_1^{\dagger}$ denotes the complement of $\Omega_1$. Define a tensor $\bP\in[0,1]^{d_1\times\ldots\times d_k}$ such that 
\begin{eqnarray*}
P(i_1,\ldots,i_k)=
\begin{cases}
\frac{n}{d_1\ldots d_k},& \textrm{if } (i_1,\ldots,i_k)\in\Omega_1\\
0,& \textrm{otherwise}.
\end{cases}
\end{eqnarray*}
Then, random tensor $\widetilde\bB$ is defined as
\begin{eqnarray*}
\widetilde B(i_1,\ldots,i_k)=
\begin{cases}
\frac{B(i_1,\ldots,i_k)}{P(i_1,\ldots,i_k)},& \textrm{with probability } P(i_1,\ldots,i_k)\\
0,& \textrm{with probability } 1-P(i_1,\ldots,i_k),
\end{cases}
\end{eqnarray*}
where we followed the convention $0/0=0$. Clearly, $\widetilde\bB-\bB$ has the same distribution as $\hat{\bA}^{\rm SPA}_{\Omega_1}-\bA_{\Omega_1}$. To apply Theorem~\ref{thm:conineq}, we observe that
$$
\nu=C_3t\max\bigg\{\frac{nd_{\max}}{d_1\ldots d_k}, k\log d_{\max}\bigg\}
$$
and 
$$
\alpha_\infty(\bB,\bP)=\max_{i_1,\ldots,i_k}\frac{|B(i_1,\ldots,i_k)|}{P(i_1,\ldots,i_k)}=\max_{i_1,\ldots,i_k}\frac{d_1\ldots d_k}{n}|B(i_1,\ldots,i_k)|\leq \frac{(d_1\ldots d_k)^{1/2}}{n}\|\bA\|_{\rm F}
$$
and
\begin{eqnarray*}
\alpha_{2,\infty}(\bB,\bP)=\max_{i_1,\ldots,i_k} \frac{|B(i_1,\ldots,i_k)|}{\sqrt{P(i_1,\ldots,i_k)}}=\max_{i_1,\ldots,i_k}\frac{(d_1\ldots d_k)^{1/2}|B(i_1,\ldots,i_k)|}{n^{1/2}}
\leq \frac{\|\bA\|_{\rm F}}{n^{1/2}}.
\end{eqnarray*}
By Theorem~\ref{thm:conineq}, with probability at least $1-d_{\max}^{-t}$, 
\begin{eqnarray*}
\|\hat{\bA}^{\rm SPA}_{\Omega_1}-\bA_{\Omega_1}\|=\|\widetilde\bB-\bB\|\leq C_1tk^3\sqrt{\frac{d_{\max}}{n}}\|\bA\|_{\rm F}+C_2tk^4\log^{k+3}(d_{\max})\frac{(d_1\ldots d_k)^{1/2}\|\bA\|_{\rm F}}{n}.
\end{eqnarray*}

\paragraph*{Upper bound of $\|\hat{\bA}^{\rm SPA}_{\Omega_2}-\bA_{\Omega_2}\|$.} Bounding $\|\hat{\bA}^{\rm SPA}_{\Omega_2}-\bA_{\Omega_2}\|$ is more involved. For $s=1,2,\ldots,\ceil{\log(d_1\ldots d_k/n)}$, define
$$
\Omega_{2,s}=\bigg\{(i_1,\ldots,i_k): |A(i_1,\ldots,i_k)|^2\in\Big[\frac{\|\bA\|_{\rm F}^2}{n}2^{-s}, \frac{\|\bA\|_{\rm F}^2}{n}2^{-s+1}\Big)\bigg\}.
$$
Clearly,
$$\Omega_2=\bigcup_{s=1}^{\ceil{\log(d_1\ldots d_k/n)}}\Omega_{2,s},$$ 
so that
$$
\|\hat{\bA}^{\rm SPA}_{\Omega_2}-\bA_{\Omega_2}\|\leq \sum_{s=1}^{\ceil{\log(d_1\ldots d_k/n)}}\big\|\hat{\bA}^{\rm SPA}_{\Omega_{2,s}}-\bA_{\Omega_{2,s}}\big\|.
$$
We now apply Theorem~\ref{thm:conineq} to bound each term on the righthand side. We follow the same strategy as before and define auxiliary tensors $\widetilde\bB_s$ and $\bB_s$ such that $\big(\bB_s\big)_{\Omega_{2,s}}=\bA_{\Omega_{2,s}}$ and $\big(\bB_s\big)_{\Omega_{2,s}^{\dagger}}={\bf 0}$.  The probability tensor $\bP_s$ is defined as
\begin{eqnarray*}
P_s(i_1,\ldots,i_k)=
\begin{cases}
\frac{n A^2(i_1,\ldots,i_k)}{\|\bA\|_{\rm F}^2},& \textrm{if } (i_1,\ldots,i_k)\in\Omega_{2,s}\\
0,& {\rm otherwise}.
\end{cases}
\end{eqnarray*}
The random tensor $\widetilde\bB_s$ is defined as 
\begin{eqnarray*}
\widetilde B_s(i_1,\ldots,i_k)=
\begin{cases}
\frac{B_s(i_1,\ldots,i_k)}{P_s(i_1,\ldots,i_k)},& \textrm{with probability } P_s(i_1,\ldots,i_k)\\
0,& \textrm{with probability } 1-P_s(i_1,\ldots,i_k).
\end{cases}
\end{eqnarray*}
Clearly, $\hat{\bA}^{\rm SPA}_{\Omega_{2,s}}-\bA_{\Omega_{2,s}}$ has the same distribution as $\widetilde\bB_s-\bB_s$. To apply Theorem~\ref{thm:conineq}, observe that
\begin{eqnarray*}
\alpha_{2,\infty}(\bB_s,\bP_s)=\max_{(i_1,\ldots,i_k)\in \Omega_{2,s}}\sqrt{\frac{B_s^2(i_1,\ldots,i_k)}{P_s(i_1,\ldots,i_k)}}=\max_{(i_1,\ldots,i_k)\in\Omega_{2,s}}\sqrt{\frac{A^2(i_1,\ldots,i_k)}{P(i_1,\ldots,i_k)}}
=\sqrt{\frac{\|\bA\|_{\rm F}^2}{n}}.
\end{eqnarray*}
Since
$$
\nu=C_1t\max\bigg\{\max_{j\in[k]}\max_{i_1,\ldots,i_{j-1},i_{j+1},\ldots,i_k}\sum_{i_j: (i_1,\ldots,i_j)\in\Omega_{2,s}} P(i_1,\ldots,i_k), k\log d_{\max}\bigg\},
$$
we obtain
\begin{eqnarray*}
\sqrt{\nu}\alpha_{\infty}(\bB_s,\bP_s)\leq C_1t^{1/2}k^{1/2}\log^{1/2}(d_{\max})\max_{(i_1,\ldots,i_k)\in\Omega_{2,s}}\frac{|A(i_1,\ldots,i_k)|}{P(i_1,\ldots,i_k)}\\
+C_1t^{1/2}\bigg(\max_{j\in[k]}\max_{i_1,\ldots,i_{j-1},i_{j+1},\ldots,i_k}\sqrt{\sum_{i_j:(i_1,\ldots,i_k)\in\Omega_{2,s}}P(i_1,\ldots,i_k)}\bigg)\max_{(i_1,\ldots,i_k)\in\Omega_{2,s}}\frac{|A(i_1,\ldots,i_k)|}{P(i_1,\ldots,i_k)}.
\end{eqnarray*}
By definition of $\Omega_{2,s}$, we have
$$
\frac{\max_{(i_1,\ldots,i_k)\in \Omega_{2,s}} P(i_1,\ldots,i_k)}{\min_{(i_1,\ldots,i_k)\in \Omega_{2,s}} P(i_1,\ldots,i_k)}\leq 2.
$$
Therefore, 
\begin{eqnarray*}
&&\bigg(\max_{j\in[k]}\max_{i_1,\ldots,i_{j-1},i_{j+1},\ldots,i_k}\sqrt{\sum_{i_j:(i_1,\ldots,i_k)\in\Omega_{2,s}}P(i_1,\ldots,i_k)}\bigg)\max_{(i_1,\ldots,i_k)\in\Omega_{2,s}}\frac{|A(i_1,\ldots,i_k)|}{P(i_1,\ldots,i_k)}\\
&\leq& \sqrt{2d_{\max}}\max_{(i_1,\ldots,i_k)\in\Omega_{2,s}}\frac{|A(i_1,\ldots,i_k)|}{\sqrt{P(i_1,\ldots,i_k)}}.
\end{eqnarray*}
By the fact $P(i_1,\ldots,i_k)=\frac{n A^2(i_1,\ldots,i_k)}{\|\bA\|_{\rm F}^2}$ and $|A(i_1,\ldots,i_k)|\geq \|\bA\|_{\rm F}/(d_1\ldots d_k)^{1/2}$, we get
\begin{eqnarray*}
\sqrt{\nu}\alpha_{\infty}(\bB_s,\bP_s)&\leq& C_1k^{1/2}t^{1/2}\log^{1/2}(d_{\max})\max_{(i_1,\ldots,i_k)\in \Omega_{2,s}}\frac{\|\bA\|_{\rm F}^2}{n |A(i_1,\ldots,i_k)|}+C_2t^{1/2}d_{\max}^{1/2}\sqrt{\frac{\|\bA\|_{\rm F}^2}{n}}\\
&\leq& C_1k^{1/2}t^{1/2}\log^{1/2}(d_{\max})\frac{(d_1\ldots d_k)^{1/2}\|\bA\|_{\rm F}}{n}+C_2t^{1/2}\sqrt{\frac{\|\bA\|_{\rm F}^2d_{\max}}{n}}.
\end{eqnarray*}
By Theorem~\ref{thm:conineq}, with probability at least $1-d_{\max}^{-t}$,
\begin{eqnarray*}
\|\hat{\bA}^{\rm SPA}_{\Omega_{2,s}}-\bA_{\Omega_{2,s}}\|\leq C_1t^2k^3\sqrt{\frac{\|\bA\|_{\rm F}^2d_{\max}}{n}}+C_2 t^2k^4\log^{k+3}(d_{\max})\frac{(d_1\ldots d_k)^{1/2}\|\bA\|_{\rm F}}{n}.
\end{eqnarray*}
By taking a uniform bound for all $s=1,2,\ldots,\ceil{\log(d_1\ldots d_k/n)}$, we conclude that with probability at least $1-k\log(d_{\max})d_{\max}^{-t}$, 
$$
\|\hat{\bA}^{\rm SPA}_{\Omega_{2}}-\bA_{\Omega_{2}}\|\leq C_1t^2k^4\log(d_{\max})\sqrt{\frac{\|\bA\|_{\rm F}^2d_{\max}}{n}}+C_2 t^2k^5\log^{k+4}(d_{\max})\frac{(d_1\ldots d_k)^{1/2}\|\bA\|_{\rm F}}{n}.
$$
\paragraph*{Finalize the proof of Lemma~\ref{lemma:sparsify}.} Put the above bounds together, we end up with, for any $t>1$, 
\begin{eqnarray*}
\|\hat{\bA}^{\rm SPA}-\bA\|\leq C_1t^2k^4\log(d_{\max})\sqrt{\frac{\|\bA\|_{\rm F}^2d_{\max}}{n}}+C_2 t^2k^5\log^{k+4}(d_{\max})\frac{(d_1\ldots d_k)^{1/2}\|\bA\|_{\rm F}}{n}
\end{eqnarray*}
which holds with probability at least $1-\big(1+k\log d_{\max} \big)d_{\max}^{-t}=1-d_{\max}^{-t+\log(k\log d_{\max})}$.
\end{proof}

\subsection{Proof of Theorem~\ref{thm:conineq}}
We begin with symmetrization \citep[see, e.g.,][]{yuan2016tensor} and obtain for any $t>0$,
$$
\PP\Big(\|\hat{\bA}-\bA\|\geq t\Big)\leq 4\PP\Big(\|\bepsilon\odot\hat{\bA}\|\geq 2t\Big)+4\exp\Big(\frac{-t^2/2}{\alpha_{2,\infty}^2(\bA,\bP)+t\alpha_{\infty}(\bA,\bP)/3}\Big)
$$
where $\beps\in\RR^{d_1\times \ldots\times d_k}$ is a random tensor with i.i.d. Rademacher entries, and
$$
\alpha_{\infty}(\bA,\bP)=\max_{i_j\in[d_j], j=1,\ldots,k}\frac{A(i_1,\ldots,i_k)}{P(i_1,\ldots,i_k)}
$$
and
$$
\alpha_{2,\infty}(\bA,\bP)=\max_{i_j\in[d_j], j\in[k]}\bigg(\frac{A^2(i_1,\ldots,i_k)}{P(i_1,\ldots,i_k)}\bigg)^{1/2}.
$$
The $\odot$ operator stands for entrywse multiplication, that is
$$\big(\eps\odot\hat A\big)(i_1,\ldots,i_k)=\eps(i_1,\ldots,i_k)\hat{A}(i_1,\ldots,i_k).$$
By definition, the operator norm $\|\beps\odot\hat\bA\|$ is given by
$$
\|\beps\odot\hat\bA\|=\underset{\bu_j\in\RR^{d_j}, \|\bu_j\|_{\ell_2}\leq 1, 1\leq j\leq k}{\sup}\ \big<\beps\odot\hat{\bA}, \bu_1\otimes\ldots\otimes \bu_k\big>.
$$
We begin with the discretization of $\ell_2$-norm balls. For each $j=1,\ldots,k$, define
$$
\mathfrak{B}_{m_j,d_j}=\big\{0,\pm1,\pm2^{-1/2},\ldots,\pm2^{-m_j/2}\big\}^{d_j}\bigcap \big\{\bu\in\RR^{d_j}: \|\bu\|_{\ell_2}\leq 1\big\}
$$
where $m_j=2\big(\ceil{\log_2d_j}+3\big)$. 
Define the ``digitalization" operator $\bD_s$ which zeros out the entries of $\bA$ whose absolute value is not $2^{-s/2}$. Then,
$$
\bD_s(\bA)=\sum_{i_1,\ldots,i_k} {\bf 1}\big\{\big|\langle\bA, \be_{i_1}\otimes\ldots\otimes\be_{i_k}\rangle\big|=2^{-s/2}\big\}A(i_1,\ldots,i_k)\be_{i_1}\otimes\ldots\otimes\be_{i_k}
$$
where we denote by $\be_{i_j}$ the canonical basis vectors in $\RR^{d_j}$. Clearly, for all $\bu_j\in\mathfrak{B}_{m_j,d_j}$, 
$$
\big<\bu_1\otimes\ldots\otimes \bu_k, \beps\odot\hat\bA\big>=\sum_{s=1}^{m_1+\ldots+m_k}\big<\bD_s\big(\bu_1\otimes\ldots\otimes \bu_k\big),\beps\odot\hat\bA\big>.
$$
For a subset $\calT\subset [d_1]\times\ldots\times[d_k]$, the aspect ratio $\mu_\calT$ is defined by
$$ 
\mu_\calT:=\max_{\ell=1,\ldots,k}\max_{i_j: j\in [k]\setminus \ell}\  {\rm Card}\big(\big\{i_\ell: (i_1,\ldots,i_k)\in\calT\big\}\big).
$$
Define the sampling locations
$$
\Omega=\big\{(i_1,\ldots,i_k): \Delta(i_1,\ldots,i_k)=1\big\}
$$
and the associated sampling operator 
$$
\calP_\Omega(\bA)=\sum_{i_1,\ldots,i_k}{\bf 1}\big((i_1,\ldots,i_k)\in\Omega\big)A(i_1,\ldots,i_k)\be_{i_1}\otimes\ldots\otimes\be_{i_k}.
$$

We shall now make use of the following version of the Chernoff bound:
\begin{lemma}\label{lemma:chernoff}
Let $X_1,\ldots,X_n$ be independent binary random variables such that $\PP(X_j=1)=p_j\in[0,1], j=1,\ldots,n$. Then, for any $t\geq 0$, 
$$
\PP\bigg(\sum_{j=1}^n\big(X_j-p_j\big)\geq 2t\sqrt{\sum_{j=1}^np_j(1-p_j)}\bigg)\leq e^{-t^2}.
$$
\end{lemma}
Lemma~\ref{lemma:chernoff} is fairly standard and we include its proof in the Appendix for completeness.

By Lemma~\ref{lemma:chernoff}, there exists an absolute constant $C>0$ such that for all $\alpha\geq 1$,
$$
\PP\Big(\mu_\Omega\geq C\alpha\max\Big\{\beta(\bP),k\log d_{\max}\Big\}\Big)\leq d_{\max}^{-\alpha}
$$
where 
$$
\beta(\bP)=\max_{j=1,\ldots,k}\ \max_{i_1,\ldots,i_{j-1}, i_{j+1},\ldots,i_k}\sum_{i_j=1}^{d_j} P(i_1,\ldots,i_k)
$$
and $d_{\max}:=\max_{1\leq j\leq k}d_j$. 
Denote the above event by $\calE_1$ with $\PP(\calE_1)\geq 1-d_{\max}^{-\alpha}$. The rest of our analysis is conditioned on event $\calE_1$. Observe that 
$$
\big<\bu_1\otimes\ldots\otimes\bu_k, \beps\odot\hat{\bA}\big>=\sum_{s=1}^{m_1+\ldots+m_k}\big<\calP_\Omega\big(\bD_s(\bu_1\otimes\ldots\otimes \bu_k)\big), \beps\odot \hat{\bA}\big>.
$$
For $\bu_j\in\mathfrak{B}_{m_j,d_j}$, let $\calA_{b_j}=\big\{i_j: \big|u_j(i_j)\big|=2^{-b_j/2}\big\}$ for $j=1,\ldots,k$. Then, we write
$$
\bD_s(\bu_1\otimes\ldots\otimes \bu_k)=\sum_{(b_1,\ldots,b_k): b_1+\ldots+b_k=s} \calP_{\calA_{b_1}\times\ldots\times\calA_{b_k}}\bD_s\big(\bu_1\otimes\ldots\otimes\bu_k\big).
$$
By definition of $\mu_\Omega$, on event $\calE_1$, there exist $\tilde{\calA}_{b_1}\subset \calA_{b_1}, \ldots, \tilde{\calA}_{b_k}\subset \calA_{b_k}$ such that 
$$
\big(\calA_{b_1}\times\ldots\times\calA_{b_k}\big)\cap \Omega=\big(\tilde{\calA}_{b_1}\otimes \ldots\otimes \tilde{\calA}_{b_k}\big)\cap \Omega
$$
and
$$
{\rm Card}^2(\tilde{\calA}_{b_j})\leq \mu_{\Omega}\prod_{j=1}^k{\rm Card}(\tilde{\calA}_{b_j}),\quad j=1,2,\ldots,k.
$$
We conclude with
$$
\big<\bD_s(\bu_1\otimes\ldots\otimes\bu_k), \beps\odot\hat{\bA}\big>=\sum_{s=1}^{m_1+\ldots+m_k}\sum_{b_1+\ldots+b_s=s}\big<\calP_{\tilde{\calA}_{b_1}\times\ldots\times\tilde\calA_{b_k}}\bD_s(\bu_1\otimes\ldots\otimes\bu_k),\beps\odot\hat{\bA}\big>.
$$
Given $\Omega$, we define the balanced version of digitalization operator
$$
\widetilde{\bD}_s(\bu_1\otimes\ldots\otimes\bu_k)=\sum_{(b_1,\ldots,b_k):b_1+\ldots+b_k=s}\calP_{\tilde{\calA}_{b_1}\times\ldots\times\tilde{\calA}_{b_k}}\bD_s\big(\bu_1\otimes\ldots\otimes\bu_k\big)
$$
where $\tilde\calA_j$ are defined as above. Then, $\calP_\Omega\bD_s(\bu_1\otimes\ldots\bu_k)=\calP_\Omega\widetilde\bD_s(\bu_1\otimes\ldots\bu_k)$. Given $\Omega$, define
$$
\mathfrak{B}_{\Omega,m_{\star}}:=\Big\{\sum_{0\leq s\leq m_{\star}}\widetilde{\bD}_s(\bu_1\otimes\ldots\bu_k)+\sum_{m_{\star}<s\leq m^{\star}}\bD_s(\bu_1\otimes\ldots\otimes\bu_k): \bu_j\in\mathfrak{B}_{m_j,d_j}, j=1,\ldots,k\Big\}
$$
for any $0<m_{\star}\leq m^{\star}\leq \sum_{j=1}^km_j$. Conditioned on $\calE_1$, we shall focus on $\{\Omega: \mu_\Omega\leq \nu\}$ where $\nu=C\alpha\max\big\{\beta(\bP),k\log d_{\max}\big\}$. Denote $\mathfrak{B}_{\nu,m_{\star}}^{\star}=\bigcup_{\mu_\Omega\leq \nu}\mathfrak{B}_{\Omega,m_{\star}}^{\star}$. Following an identical argument as that in \cite{yuan2016tensor}, we get
$$
\big\|\beps\odot\hat\bA \big\|\leq 2^k \max_{\bY\in\mathfrak{B}^{\star}_{\nu,m_{\star}}}\langle\bY, \beps\odot\hat\bA\rangle.
$$
The entropy number of $\mathfrak{B}^{\star}_{\nu,m_{\star}}$ plays an essential role in bounding $\max_{\bY\in\mathfrak{B}^{\star}_{\nu,m_{\star}}}\langle\bY, \bX\rangle$. Observe that $\mathfrak{B}^{\star}_{\nu,m_{\star}}\subset \mathfrak{B}_{m_1,d_1}\times\ldots\times\mathfrak{B}_{d_k,m_k}$ and
\begin{eqnarray*}
{\rm Card}\big(\mathfrak{B}_{m_j,d_j}\big)&\leq& \prod_{k=0}^{m_j}{d_j \choose 2^k\wedge d_j}2^{2^k\wedge d_j}\\
&\leq& \prod_{k=0}^{m_j}\exp\Big((2^k\wedge d_j)\big(\log 2+1+(\log d_j/2^k)_+\big)\Big)\\
&\leq& \exp\Big(d_j\sum_{\ell=1}^{\infty}2^{-\ell}\big(\log2+1+\log(2^\ell)\big)\Big)\\
&\leq& \exp\big(21d_j/4\big),
\end{eqnarray*}
which implies that 
$$
\log{\rm Card}\Big(\mathfrak{B}^{\star}_{\nu,m_{\star}}\Big)\leq \frac{21}{4}\big(d_1+\ldots+d_k\big).
$$
See \cite{yuan2016tensor} for more details. More precise characterizations of ${\rm Card}(\mathfrak{B}^{\star}_{\nu,m_{\star}})$ can also be derived. For any $0\leq q\leq s\leq m_{\star}$, define
$$
\mathfrak{D}_{\nu,s,q}=\big\{\bD_s(\bY): \bY\in\mathfrak{B}^{\star}_{\nu,m_{\star}}, \|\bD_\bs(\bY)\|_{\ell_2}^2\leq 2^{q-s}\big\}.
$$
\begin{lemma}\label{lemma:Dvsqent}
Let $\nu\geq 1$. For all $0\leq q\leq s\leq m_{\star}$, the following bound holds
$$
\log{\rm Card}(\mathfrak{D}_{\nu,s,q})\leq qs^k\log 2+2k^2s^{k}\sqrt{\nu 2^q} L\big(\sqrt{\nu 2^q}, d_{\max}s^{k/2}\big)
$$
where $L(x,y)=\max\big\{1, \log(ey/x)\big\}$.
\end{lemma}
We  write
\begin{eqnarray*}
\|\beps\odot\hat\bA\|&\leq& 2^k\max_{\bY\in\mathfrak{B}^{\star}_{\nu,m_{\star}}}\big<\bY, \beps\odot\hat\bA\big>\\
&=&2^k\max_{\bY\in\mathfrak{B}^{\star}_{\nu,m_{\star}}}\bigg(\sum_{0\leq s\leq m_{\star}}\big<\bD_s\big(\bY\big), \beps\odot\hat\bA\big>+\big<\bS_\star(\bY),\beps\odot\hat\bA\big>\bigg)
\end{eqnarray*}
where $\bS_\star(\bY)=\sum_{s>m_{\star}}\bD_s(\bY)$. The actual value of $m_{\star}$ is to be determined later. 

\paragraph*{Upper bound of $\big|\big<\bD_s(\bY),\beps\odot\hat\bA\big>\big|$.} Recall the definition of $\mathfrak{D}_{\nu,s,q}$ and that
$$2^{-s}\leq \|\bD_s(\bY)\|_{\ell_2}^2\leq 1,$$
we can write 
$$\bD_s(\bY)\in\bigcup_{q=1}^s\big(\mathfrak{D}_{\nu,s,q}\setminus \mathfrak{D}_{\nu,s,q-1}\big).$$
Then
$$
\max_{\bY\in\mathfrak{B}_{\nu,m_{\star}}^{\star}} \big<\bD_s(\bY), \beps\odot\hat\bA\big>=\max_{1\leq q\leq s}\ \max_{\bY_{s,q}\in \mathfrak{D}_{\nu,s,q}\setminus \mathfrak{D}_{\nu,s,q-1}}\big<\bY_{s,q}, \beps\odot\hat{\bA}\big>.
$$
Observe that
$$
\big<\bY_{s,q}, \beps\odot\hat{\bA}\big>=\sum_{i_j\in[d_j], j=1,\ldots,k} \frac{\Delta(i_1,\ldots,i_k)}{P(i_1\ldots i_k)}\eps(i_1,\ldots,i_k)A(i_1,\ldots,i_k)Y_{s,q}(i_1,\ldots,i_k),
$$
where $\bDelta$ is a binary random tensor and $\beps$ is a Rademacher random tensor. Both of them have i.i.d. entries. 
By definition of $\bY_{s,q}$ and $\mathfrak{D}_{\nu,s,q}$,  we have $\max_{i_1,\ldots,i_k}|Y_{s,q}(i_1,\ldots,i_k)|\leq 2^{-s/2}$. Moreover, 
\begin{eqnarray*}
{\rm Var}\big(\big<\bY_{s,q}, \beps\odot\hat{\bA}\big>\big)=\sum_{i_j\in[d_j], j=1,\ldots,k}\frac{A^2(i_1,\ldots,i_k)}{P(i_1,\ldots,i_k)}Y_{s,q}^2(i_1,\ldots,i_k).
\end{eqnarray*}
Since $\|\bY_{s,q}\|_{\rm F}^2\leq 2^{q-s}$, we obtain
$$
{\rm Var}\big(\big<\bY_{s,q}, \beps\odot\hat{\bA}\big>\big)\leq 
\max_{i_j\in[d_j], j\in[k]} \frac{A^2(i_1,\ldots,i_k)}{P(i_1,\ldots,i_k)}\|\bY_{s,q}\|_{\rm F}^2\leq 2^{q-s}\max_{i_j\in[d_j], j\in[k]} \frac{A^2(i_1,\ldots,i_k)}{P(i_1,\ldots,i_k)}.
$$
Recall the definition of $\alpha_{\infty}(\bA,\bP)$ and $\alpha_{2,\infty}(\bA,\bP)$.
By Bernstein inequality for sum of bounded random variables, there exist absolute constants $C_0,C_1,C_2>0$ such that
\begin{eqnarray*}
\PP\Big(\big|\big<\bY_{s,q}, \beps\odot \hat\bA\big>\big|\geq t\Big)
\leq \exp\bigg(-\frac{C_0t^2}{C_12^{q-s}\alpha_{2,\infty}^2(\bA,\bP)+C_22^{-s/2}t\alpha_\infty(\bA,\bP)}\bigg)
\end{eqnarray*}
for any $t>0$. By the union bound and Lemma~\ref{lemma:Dvsqent}, we get
\begin{eqnarray*}
&&\PP\Big(\max_{\bY_{s,q}\in\mathfrak{D}_{\nu,s,q}}\big|\big<\bY_{s,q}, \beps\odot \hat\bA\big>\big|\geq t\Big)\\
&\leq& {\rm Card}\big(\mathfrak{D}_{\nu,s,q}\big)\exp\bigg(-\frac{C_0t^2}{C_12^{q-s}\alpha_{2,\infty}^2(\bA,\bP)+C_22^{-s/2}t\alpha_\infty(\bA,\bP)}\bigg)\\
&\leq& \exp\bigg(21\big(\sum_{j=1}^kd_j\big)/4-\frac{C_0t^2}{C_12^{q-s}\alpha_{2,\infty}^2(\bA,\bP)}\bigg)\\
&&+\exp\bigg(qs^k\log 2+2k^2s^k\sqrt{\nu2^q}L\big(\sqrt{\nu2^q}, d_{\max}s^{k/2}\big)-\frac{C_02^{s/2}t^2}{C_2t\alpha_\infty(\bA,\bP)}\bigg).
\end{eqnarray*}
Recall that
$$0\leq q\leq s\leq m_{\star}\lesssim k\log d_{\max}$$
and 
$$L\big(\sqrt{\nu2^q}, d_{\max}s^{k/2}\big)\lesssim {\frac{k}{2}}\log d_{\max}.$$
For large enough constants $C_3,C_4>0$, by choosing $t>0$ such that
\begin{eqnarray*}
t\geq C_32^{(q-s)/2}\Big(\sum_{j=1}^kd_j\Big)^{1/2}\alpha_{2,\infty}(\bA,\bP)+C_4k^3\log^{k+1}d_{\max}\sqrt{\nu}2^{q-s}\alpha_\infty(\bA,\bP),
\end{eqnarray*}
we get for any $0\leq q\leq s\leq m_{\star}$,
\begin{eqnarray*}
\PP\Big(\max_{\bY_{s,q}\in\mathfrak{D}_{\nu,s,q}}\big|\big<\bY_{s,q}, \beps\odot \hat\bA\big>\big|\geq t\Big)\leq \exp\bigg(-\frac{C_0t^2}{C_12^{q-s}\alpha_{2,\infty}^2(\bA,\bP)}\bigg)+\exp\bigg(-\frac{C_02^{s/2}t}{C_2\alpha_\infty(\bA,\bP)}\bigg).
\end{eqnarray*}
By making the above bound uniform over all pairs $0\leq q\leq s\leq m_{\star}$, we obtain
\begin{eqnarray*}
\PP\Big(\max_{\bY\in\mathfrak{B}_{\nu,m_{\star}}^{\star}}\Big|\sum_{0\leq s\leq m_{\star}}\big<\bD_s(\bY), \beps\odot\hat{\bA}\big>\Big|\geq (m_{\star}+1)t\Big)\leq {m_{\star}+1\choose 2}\exp\bigg(-\frac{C_0t^2}{C_1\alpha_{2,\infty}^2(\bA,\bP)}\bigg)\\
+{m_{\star}+1\choose 2}\exp\bigg(-\frac{C_0t}{C_2\alpha_\infty(\bA,\bP)}\bigg).
\end{eqnarray*}

\paragraph*{Upper bound of $\max_{\bY\in\mathfrak{B}_{\nu,m_{\star}}^{\star}}\big|\big<\bS_{\star}(\bY), \beps\odot\hat\bA\big>\big|$.} For notation simplicity, we write $\bS_{\star}$ in short for $\bS_{\star}(\bY)$. We apply Bernstein inequality to 
$$
\big<\bS_{\star}, \beps\odot\hat\bA\big>=\sum_{i_j\in[d_j], j=1,\ldots,k} \frac{\Delta(i_1,\ldots,i_k)}{P(i_1,\ldots,i_k)}\eps(i_1,\ldots,i_k)A(i_1,\ldots,i_k)S_{\star}(i_1,\ldots,i_k).
$$
Clearly, $\big|S_{\star}(i_1,\ldots,i_k)\big|\leq 2^{-m_{\star}/2}$. Meanwhile,
\begin{eqnarray*}
\Var\big(\big<\bS_{\star}, \beps\odot\hat\bA\big>\big)=\sum_{i_j\in[d_j], j=1,\ldots,k}\frac{A^2(i_1,\ldots,i_k)}{P(i_1,\ldots,i_k)}S_{\star}^2(i_1,\ldots,i_k).
\end{eqnarray*}
Following an identical approach as previously, we show that
$$
\Var\big(\big<\bS_{\star}, \beps\odot\hat\bA\big>\big)\leq \alpha_{2,\infty}^2(\bA,\bP).
$$
By Bernstein inequality and the union bound
\begin{eqnarray*}
&&\PP\Big(\max_{\bY\in\mathfrak{B}_{\nu,m_{\star}}^{\star}}\big|\big<\bS_{\star}(\bY), \beps\odot\hat\bA\big>\big|\geq t\Big)\\
&\leq& {\rm Card}\big(\mathfrak{B}_{\nu,m_{\star}}^{\star}\big)\exp\bigg(-\frac{C_0t^2}{C_1\alpha_{2,\infty}^2(\bA,\bP)+C_22^{-m_{\star}/2}t\alpha_{\infty}(\bA,\bP)}\bigg)\\
&\leq& \exp\bigg(21\sum_{j=1}^kd_j/4-\frac{C_0t^2}{C_1\alpha_{2,\infty}^2(\bA,\bP)}\bigg)+\exp\bigg(21\sum_{j=1}^kd_j/4-\frac{C_02^{m_{\star}/2}t}{C_2\alpha_{\infty}(\bA,\bP)}\bigg)
\end{eqnarray*}
for some absolute constants $C_0,C_1,C_2>0$. For large enough constants $C_3,C_4>0$, by choosing $t$ such that
$$
t\geq C_3\Big(\sum_{j=1}^kd_j\Big)^{1/2}\alpha_{2,\infty}(\bA,\bP)+C_4\Big(\sum_{j=1}^kd_j\Big)2^{-m_{\star}/2}\alpha_\infty(\bA,\bP),
$$
we obtain
\begin{eqnarray*}
\PP\Big(\max_{\bY\in\mathfrak{B}_{\nu,m_{\star}}^{\star}}\big|\big<\bS_{\star}(\bY), \beps\odot\hat\bA\big>\big|\geq t\Big)\leq \exp\bigg(-\frac{C_0t^2}{C_1\alpha_{2,\infty}^2(\bA,\bP)}\bigg)+\exp\bigg(-\frac{C_02^{m_{\star}/2}t}{C_2\alpha_{\infty}(\bA,\bP)}\bigg).
\end{eqnarray*}

\paragraph*{Finalize the proof of Theorem~\ref{thm:conineq}.} Combining above bounds, we conclude that if for large enough constants $C_3,C_4,C_5>0$ such that
\begin{eqnarray*}
t\geq C_3\Big(\sum_{j=1}^kd_j\Big)^{1/2}\alpha_{2,\infty}(\bA,\bP)+C_4k^3\log^{k+1}(d_{\max})\sqrt{\nu}\alpha_{\infty}(\bA,\bP)+C_5\Big(\sum_{j=1}^kd_j\Big)2^{-m_{\star}/2}\alpha_{\infty}(\bA,\bP).
\end{eqnarray*}
Thus
\begin{eqnarray*}
\PP\Big(\|\beps\odot\hat\bA\|\geq (m_{\star}+2)t\Big)&\leq& \left({m_{\star}+1\choose 2}+1\right)\exp\bigg(-\frac{C_0t^2}{C_1\alpha_2^2(\bA,\bP)}\bigg)\\
&&+\left({m_{\star}+1\choose 2}+1\right)\exp\bigg(-\frac{C_0t}{C_2\alpha_{\infty}(\bA,\bP)}\bigg).
\end{eqnarray*}
Recall that $\nu=C_1\alpha\max\big\{\beta(\bP), k\log d_{\max}\big\}$ and $m_{\star}\leq \sum_{j=1}^k2\Big(\ceil{\log_2 d_j}+3\Big)$. By choosing $m_{\star}$ large enough such that $2^{-m_{\star}/2}\Big(\sum_{j=1}^kd_j\Big)\leq \sqrt{\nu}$, we conclude that for any $\gamma>0$ such that
$$
t\geq C_3\bigg(\Big(\sum_{j=1}^kd_j\Big)^{1/2}+\gamma k\log d_{\max}\bigg)\alpha_{2,\infty}(\bA,\bP)+C_4\gamma k^3\log^{k+2}(d_{\max})\sqrt{\nu}\alpha_\infty(\bA,\bP).
$$
It follows immediately, by adjusting the constant $C_3$, that
$$
\PP\Big(\|\hat\bA-\bA\|\geq t\Big)\leq 2d_{\max}^{-\gamma}.
$$

\subsection{Proof of Theorem~\ref{thm:SVD}}
It suffices to prove the upper bound of $\|\hat\bM_j-\bM_j\|$ where $\bM_j=\calM_j(\bA)$ and $\hat{\bM}_j=\calM_j(\hat{\bA}^{\rm SPA})$. Without loss of generality, let $j=1$. Recall the notation $d_{-1}=d_2\ldots d_k$. 
 By denoting $\bE_{i_1(i_2\ldots i_k)}\in\RR^{d_1\times d_{-1}}$ the canonical basis matrices of $\RR^{d_1\times d_{-1}}$ that is $\bE_{i_1(i_2\ldots i_k)}$ has exactly value $1$ on the $(i_1,i_2\ldots i_k)$ position and all $0$'s elsewhere. Then,
$$
\hat\bM_j-\bM_j=\sum_{i_j\in[d_j], 1\leq j\leq k} \Big(\frac{A(i_1,\ldots,i_k)\Delta(i_1,\ldots,i_k)}{P(i_1,\ldots,i_k)}-A(i_1,\ldots,i_k)\Big)\bE_{i_1(i_2\ldots i_k)}
$$
where $\PP\big(\Delta(i_1,\ldots,i_k)=1\big)=P(i_1,\ldots,i_k)$. We shall apply the matrix Bernstein inequality to bound the sum of random matrices for $\hat\bM_j-\bM_j$. Denote the locations of small entries by 
$$
\Omega_1:=\big\{(i_1,\ldots,i_k): \|A(i_1,\ldots,i_k)\|\leq \|\bA\|_{\rm F}/(d_1\ldots d_k)^{1/2}\big\}\subset [d_1]\times\ldots\times[d_k]
$$
moderate entries by
$$
\Omega_2:=\big\{(i_1,\ldots,i_k): \|A(i_1,\ldots,i_k)\|/\|\bA\|_{\rm F}\in\big(1/(d_1\ldots d_k)^{1/2}, 1/n^{1/2}\big)\big\}\subset [d_1]\times\ldots\times[d_k]
$$
and  large entries by 
$$
\Omega_3:=\big\{(i_1,\ldots,i_k): \|A(i_1,\ldots,i_k)\|\geq \|\bA\|_{\rm F}/n^{1/2}\big\}\subset [d_1]\times\ldots\times[d_k].
$$
Recall that $P(i_1,\ldots,i_k)=1$ for $(i_1,\ldots,i_k)\in\Omega_3$. Then, for any $(i_1,\ldots,i_k)\in \Omega_1\cup \Omega_2$, we have
$$
\Big\|\Big(\frac{A(i_1,\ldots,i_k)\Delta(i_1,\ldots,i_k)}{P(i_1,\ldots,i_k)}-A(i_1,\ldots,i_k)\Big)\bE_{i_1(i_2\ldots i_k)}\Big\|\leq \max_{i_j\in[d_j], 1\leq j\leq k} \bigg|\frac{A(i_1,\ldots,i_k)}{P(i_1,\ldots,i_k)}\bigg|.
$$
Moreover, 
\begin{eqnarray*}
&&\Big\|\sum_{i_j\in[d_j], 1\leq j\leq k}\EE\Big(\frac{A(i_1,\ldots,i_k)\Delta(i_1,\ldots,i_k)}{P(i_1,\ldots,i_k)}-A(i_1,\ldots,i_k)\Big)^2\bE_{i_1(i_2\ldots i_k)}\bE_{i_1(i_2\ldots i_k)}^\top\Big\|\\
&\leq& \max_{1\leq i_1\leq d_1}\sum_{i_j\in[d_j], 2\leq j\leq k}\frac{A^2(i_1,\ldots,i_k)\big(1-P(i_1,\ldots,i_k)\big)}{P(i_1,\ldots,i_k)}\\
&\leq& \max_{1\leq i_1\leq d_1}\sum_{i_j\in[d_j], j\geq 2, (i_1,\ldots,i_k)\in\Omega_1\cup \Omega_2}\ 
\frac{A^2(i_1,\ldots,i_k)}{P(i_1,\ldots,i_k)}.
\end{eqnarray*}
Similarly,
\begin{eqnarray*}
&&\Big\|\sum_{i_j\in[d_j], 1\leq j\leq k}\EE\Big(\frac{A(i_1,\ldots,i_k)\Delta(i_1,\ldots,i_k)}{P(i_1,\ldots,i_k)}-A(i_1,\ldots,i_k)\Big)^2\bE_{i_1(i_2\ldots i_k)}^\top\bE_{i_1(i_2\ldots i_k)}\Big\|\\
&\leq& \max_{i_j\in[d_j], 2\leq j\leq k} \sum_{i_1=1}^{d_1}\frac{A^2(i_1,\ldots,i_k)\big(1-P(i_1,\ldots,i_k)\big)}{P(i_1,\ldots,i_k)}\\
&\leq& \max_{i_j\in[d_j],2\leq j\leq k} \sum_{i_1\in [d_1], (i_1,\ldots,i_k)\in\Omega_1\cup \Omega_2}\frac{A^2(i_1,\ldots,i_k)}{P(i_1,\ldots,i_k)}.
\end{eqnarray*}
Observe that if $(i_1,\ldots,i_k)\in \Omega_1$, then
$$
\Big|\frac{A(i_1,\ldots,i_k)}{P(i_1,\ldots,i_k)}\Big|=\frac{(d_1\ldots d_k)}{n}|A(i_1,\ldots,i_k)|\leq \frac{(d_1\ldots d_k)^{1/2}}{n}\|\bA\|_{\rm F}
$$
and
$$
\frac{A^2(i_1,\ldots,i_k)}{P(i_1,\ldots,i_k)}=\frac{(d_1\ldots d_k)A^2(i_1,\ldots,i_k)}{n}\leq \frac{\|\bA\|_{\rm F}^2}{n}.
$$
Similarly, if $(i_1,\ldots,i_k)\in\Omega_2$, then
$$
\Big|\frac{A(i_1,\ldots,i_k)}{P(i_1,\ldots,i_k)}\Big|=\frac{\|\bA\|_{\rm F}^2}{n|A(i_1,\ldots,i_k)|}\leq \frac{(d_1\ldots d_k)^{1/2}}{n}\|\bA\|_{\rm F}
$$
and
$$
\frac{A^2(i_1,\ldots,i_k)}{P(i_1,\ldots,i_k)}= \frac{\|\bA\|_{\rm F}^2}{n}.
$$
By matrix Bernstein inequality \citep{tropp2012user}, for any $t\geq 0$, with probability at least $1-e^{-t}$ that
\begin{eqnarray*}
\big\|\hat\bM_j-\bM_j \big\|\leq 2\|\bA\|_{\rm F}\bigg(\sqrt{\frac{d_2d_3\ldots d_k(t+k\log d_{\max})}{n}}+\frac{(d_1\dots d_k)^{1/2}(t+k\log d_{\max})}{n}\bigg).
\end{eqnarray*}
Since $\hat\bM_j=\bM_j+\big(\hat\bM_j-\bM_j\big)$, the claim follows directly from Davis-Kahan Thoerem as in (\ref{eq:Davis-Kahan}). 

\subsection{Proof of Theorem~\ref{thm:SVD2}}
Theorem~\ref{thm:SVD2} is an immediate consequence of the following concentration bound. 
\begin{lemma}\label{lem:SVD2}
Let $\bU_j^{(r_j)}$ be the $r_j$ leading left singular vectors of $\calM_j(\bA)$, and $\hat\bU_j^{(r_j)}$ be the output from Algorithm~\ref{algo:HOSVD_SPA}. There exist constants $C_1, C_2>0$ depending on $k$ only such that if 
$$
n\geq C_1(d_1\ldots d_k)^{1/2}(t+\log d_{\max}),
$$
then for any $t\geq 0$, the following bound holds with probability at least $1-e^{-t}$:
\begin{eqnarray*}
\big\|\hat\bU_j^{(r_j)}\big(\hat\bU_j^{(r_j)}\big)^{\top}-\bU_j^{(r_j)}\big(\bU_j^{(r_j)}\big)^{\top}\big\|\\
\leq C_2\frac{\|\bA\|_{\rm F}}{\bar{g}_{r_j}(\bM_j\bM_j^{\top})}\bigg(\sigma_{\max}(\bM_j)\sqrt{\frac{d_j(t+\log d_{\max})}{n}}+\|\bA\|_{\rm F}\frac{(d_1\ldots d_k)^{1/2}(t+\log d_{\max})}{n}\bigg).
\end{eqnarray*}
\end{lemma}
\begin{proof}[Proof of Lemma~\ref{lem:SVD2}]
With out loss of generality, we assume $j=1$ without loss of generality. In this case, $\hat\bM_j^{(1)}=\calM_j(\hat{\bA}^{\rm SPA}_1), \hat\bM_j^{(2)}=\calM_j(\hat{\bA}^{\rm SPA}_2)\in \RR^{d_1\times (d_2\ldots d_k)}$. Observe that
\begin{eqnarray*}
\hat\bM_j^{(1)}\big(\hat\bM_j^{(2)}\big)^\top=\bM_j\bM_j^{\top}+\big(\hat\bM_j^{(1)}-\bM_j\big)\bM_j^{\top}+\bM_j\big(\hat\bM_j^{(2)}-\bM_j\big)^{\top}\\
+\big(\hat\bM_j^{(1)}-\bM_j\big)\big(\hat\bM_j^{(2)}-\bM_j\big)^{\top}.
\end{eqnarray*}

\paragraph*{Upper bound of $\big\|\big(\hat\bM_j^{(1)}-\bM_j\big)\big(\hat\bM_j^{(2)}-\bM_j\big)^{\top} \big\|$.} Denote by $\bZ_1=\hat{\bM}_j^{(1)}-\bM_j$. By Theorem~\ref{thm:SVD}, there exists an event $\calE_1$ with $\PP(\calE_1)\geq 1-e^{-t}$ such that on event $\calE_1$, 
$$
\|\bZ_1\|\leq C\|\bA\|_{\rm F}\bigg(\sqrt{\frac{d_2d_3\ldots d_k(t+\log d_{\max})}{n}}+\frac{(d_1\dots d_k)^{1/2}(t+\log d_{\max})}{n}\bigg).
$$
Denote by $\|\bZ_1\|_{2,\infty}$ the maximal column $\ell_2$ norm., i.e., $\|\bZ_1\|_{2,\infty}=\max_{j\in[d_2\ldots d_k]}\big\|\bZ_1\be_j\big\|_{\ell_2}$. 
Clearly, there exists a constant $C_1$ depending on $k$ only such that
\begin{eqnarray*}
\|\bZ_1\|_{2,\infty}&\leq& C_1\bigg(\max_{i_j\in[d_j], 2\leq j\leq k}\sqrt{\sum_{i_1\in[d_1]: (i_1,\ldots,i_k)\in\Omega_1\cup \Omega_2} \frac{A^2(i_1,\ldots,i_k)}{P(i_1,\ldots,i_k)} (t+\log d_{\max})}\\
&&\hskip 50pt+\max_{(i_1,\ldots,i_k)\in \Omega_1\cup \Omega_2} \bigg|\frac{A(i_1,\ldots,i_k)}{P(i_1,\ldots,i_k)} \bigg|(t+\log d_{\max})\bigg)\\
&\leq& C_1\|\bA\|_{\rm F}\bigg(\sqrt{\frac{d_1(t+\log d_{\max})}{n}}+\frac{(d_1\ldots d_k)^{1/2}(t+\log d_{\max})}{n}\bigg),
\end{eqnarray*}
which holds with probability at least $1-e^{-t}$.
Denote the above event by $\calE_3$. We shall proceed conditional on $\calE_1\cap \calE_2\cap \calE_3$. Write
$$
\bZ_1\big(\hat\bM_j^{(2)}-\bM_j\big)^{\top}=\sum_{i_j\in[d_j], 1\leq j\leq k} \Big(\frac{A(i_1,\ldots,i_k)\Delta(i_1,\ldots,i_k)}{P(i_1,\ldots,i_k)}-A(i_1,\ldots,i_k)\Big)\bZ_1\bE_{i_1(i_2\ldots i_k)}^\top
$$
which is again a sum of random matrices. Clear, for any $(i_1,\ldots,i_k)\in \Omega_1\cup \Omega_2$, 
\begin{eqnarray*}
&&\Big\|\Big(\frac{A(i_1,\ldots,i_k)\Delta(i_1,\ldots,i_k)}{P(i_1,\ldots,i_k)}-A(i_1,\ldots,i_k)\Big)\bZ_1\bE_{i_1(i_2\ldots i_k)}^\top \Big\|\\
&\leq& \max_{(i_1,\ldots,i_k)\in \Omega_1\cup \Omega_2}\Big|\frac{A(i_1,\ldots,i_k)}{P(i_1,\ldots,i_k)}\Big|\|\bZ_1\|_{2,\infty}\\
&\leq& \frac{(d_1\ldots d_k)^{1/2}}{n}\|\bA\|_{\rm F}\|\bZ_1\|_{2,\infty}.
\end{eqnarray*}
Moreover, 
\begin{eqnarray*}
&&\bigg\|\sum_{i_j\in[d_j], 1\leq j\leq k}\EE\Big(\frac{A(i_1,\ldots,i_k)\Delta(i_1,\ldots,i_k)}{P(i_1,\ldots,i_k)}-A(i_1,\ldots,i_k)\Big)^2\bZ_1\bE_{i_1(i_2\ldots i_k)}^{\top}\bE_{i_1(i_2\ldots i_k)}\bZ_1^{\top}\bigg\|\\
&\leq& \max_{i_j\in [d_j], 2\leq j\leq k}\|\bZ_1\|^2 \sum_{i_1\in [d_1]: (i_1,\ldots,i_k)\in \Omega_1\cup \Omega_2} \frac{A^2(i_1,\ldots,i_k)}{P(i_1,\ldots,i_k)}\\
&\leq& \frac{d_1\|\bA\|_{\rm F}^2}{n}\|\bZ_1\|^2.
\end{eqnarray*}
Similarly, 
\begin{eqnarray*}
&&\bigg\|\sum_{i_j\in[d_j], 1\leq j\leq k}\EE\Big(\frac{A(i_1,\ldots,i_k)\Delta(i_1,\ldots,i_k)}{P(i_1,\ldots,i_k)}-A(i_1,\ldots,i_k)\Big)^2\bE_{i_1(i_2\ldots i_k)}\bZ_1^\top\bZ_1\bE_{i_1(i_2\ldots i_k)}^{\top}\bigg\|\\
&\leq& \max_{(i_1,\ldots,i_k)\in \Omega_1\cup \Omega_2}\frac{A^2(i_1,\ldots,i_k)}{P(i_1,\ldots,i_k)}\|\bZ_1\|_{\rm F}^2\\
&\leq& \frac{d_1\|\bA\|_{\rm F}^2}{n}\|\bZ_1\|^2.
\end{eqnarray*}
By matrix Bernstein inequality, the following bound holds with probability at least $1-e^{-t}$,
\begin{eqnarray*}
&&\big\|\big(\hat\bM_j^{(1)}-\bM_j\big)\big(\hat\bM_j^{(2)}-\bM_j\big)^{\top}\big\|\\
&\leq& C\|\bA\|_{\rm F}\bigg(\sqrt{\frac{d_1(t+\log d_{\max})}{n}}\|\bZ_1\|+\frac{(d_1\ldots d_k)^{1/2}(t+\log d_{\max})}{n}\|\bZ_1\|_{2,\infty}\bigg).
\end{eqnarray*}
Denote the above event by $\calE_4$. On event $\calE_1\cap \calE_2\cap \calE_3\cap \calE_4$, if 
$$
n\geq C_1(d_1d_2\ldots d_k)^{1/2}(t+\log d_{\max}),
$$
then
\begin{eqnarray*}
&&\big\|\big(\hat\bM_j^{(1)}-\bM_j\big)\big(\hat\bM_j^{(2)}-\bM_j\big)^{\top}\big\|\\
&\leq& C_2\|\bA\|_{\rm F}^2\bigg(\frac{(d_1\ldots d_k)^{1/2}(t+\log d_{\max})}{n}+\frac{d_1^{1/2}(d_1\ldots d_k)^{1/2}(t+\log d_{\max})^{3/2}}{n^{3/2}}\bigg)\\
&\leq& C_2\|\bA\|_{\rm F}^2\frac{(d_1\ldots d_k)^{1/2}(t+\log d_{\max})}{n}.
\end{eqnarray*}

\paragraph*{Upper bound of $\big\|\bM_j\big(\hat\bM_j^{(2)}-\bM_j\big)^{\top}\big\|$.} We write
$$
\bM_j\big(\hat\bM_j^{(2)}-\bM_j\big)^{\top}=\sum_{i_j\in[d_j], 1\leq j\leq k}\bigg(\frac{A(i_1,\ldots,i_k)\Delta(i_1,\ldots,i_k)}{P(i_1,\ldots,i_k)}-A(i_1,\ldots,i_k)\bigg)\bM_j\bE_{i_1(i_2\ldots i_k)}^\top.
$$
The proof follows identically as above. Indeed, for any $(i_1,\ldots,i_k)\in \Omega_1\cup \Omega_2$,
\begin{eqnarray*}
&&\bigg\|\bigg(\frac{A(i_1,\ldots,i_k)\Delta(i_1,\ldots,i_k)}{P(i_1,\ldots,i_k)}-A(i_1,\ldots,i_k)\bigg)\bM_j\bE_{i_1(i_2\ldots i_k)}^\top\bigg\|\\
&\leq& \max_{(i_1,\ldots,i_k)\in\Omega_1\cup \Omega_2}\Big|\frac{A(i_1,\ldots,i_k)}{P(i_1,\ldots,i_k)}\Big|\max_{i_j\in[d_j], 2\leq j\leq k}\sqrt{\sum_{i_1: (i_1,\ldots,i_k)\in\Omega_1\cup \Omega_2}A^2(i_1,\ldots,i_k)}\\
&\leq& \frac{(d_1\ldots d_k)^{1/2}}{n}\Big(\frac{d_1}{n}\Big)^{1/2}\|\bA\|_{\rm F}^2.
\end{eqnarray*}
Moreover, 
\begin{eqnarray*}
&&\bigg\| \sum_{i_j\in[d_j], 1\leq j\leq k}\EE\bigg(\frac{A(i_1,\ldots,i_k)\Delta(i_1,\ldots,i_k)}{P(i_1,\ldots,i_k)}-A(i_1,\ldots,i_k)\bigg)^2\bM_j\bE_{i_1(i_2\ldots i_k)}^\top\bE_{i_1(i_2\ldots i_k)}\bM_j^\top\bigg\|\\
&\leq& \max_{i_j\in[d_j], 2\leq j\leq k}\sum_{i_1: (i_1,\ldots,i_k)\in\Omega_1\cup \Omega_2}\frac{A^2(i_1,\ldots,i_k)}{P(i_1,\ldots,i_k)}\|\bM_j\|^2\\
&\leq& \frac{d_1}{n}\|\bA\|_{\rm F}^2\sigma_{\max}^2(\bM_j).
\end{eqnarray*}
Similarly,
\begin{eqnarray*}
&&\bigg\| \sum_{i_j\in[d_j], 1\leq j\leq k}\EE\bigg(\frac{A(i_1,\ldots,i_k)\Delta(i_1,\ldots,i_k)}{P(i_1,\ldots,i_k)}-A(i_1,\ldots,i_k)\bigg)^2\bE_{i_1(i_2\ldots i_k)}\bM_j^\top\bM_j\bE_{i_1(i_2\ldots i_k)}^\top\bigg\|\\
&\leq& \bigg(\max_{i_j\in[d_j], 2\leq j\leq k}\sum_{i_1: (i_1,\ldots,i_k)\in\Omega_1\cup \Omega_2}A^2(i_1,\ldots,i_k)\bigg)\bigg(\max_{i_1\in[d_1]}\sum_{i_j\in[d_j], 2\leq j\leq k: (i_1,\ldots,i_k)\in\Omega_1\cup \Omega_2}\frac{A^2(i_1,\ldots,i_k)}{P(i_1,\ldots,i_k)}\bigg)\\
&\leq& \frac{d_1d_2\ldots d_k}{n^2}\|\bA\|_{\rm F}^4.
\end{eqnarray*}
By matrix Bernstein inequality \citep{tropp2012user}, if $n\geq C_1(d_1\ldots d_k)^{1/2}(t+\log d_{\max})$, then with probability at least $1-e^{-t}$ such that
\begin{eqnarray*}
\big\|\bM_j\big(\hat\bM_j^{(2)}-\bM_j\big)^\top\big\|\\
\leq C_2\|\bA\|_{\rm F}\bigg(\sigma_{\max}(\bM_j)\sqrt{\frac{d_1(t+\log d_{\max})}{n}}+\|\bA\|_{\rm F}\frac{(d_1\ldots d_k)^{1/2}(t+\log d_{\max})}{n}\bigg).
\end{eqnarray*}
Denote this event by $\calE_5$. Clearly, an identical bound holds for $\big\|\big(\hat\bM_j^{(1)}-\bM_j\big)\bM_j^\top\big\|$ with the same probability. Denote this event by $\calE_6$. 
\paragraph*{Finalize the proof of Theorem~\ref{thm:SVD2}.} On event $\calE_1\cap \calE_2\cap \calE_3\cap \calE_4\cap \calE_5\cap \calE_6$, if $n\geq C_1(d_1\ldots d_k)^{1/2}(t+\log d_{\max})$,
there exists a constant $C_2$ depending on $k$ only such that
\begin{eqnarray*}
\big\|\hat\bM_j^{(1)}\big(\hat\bM_j^{(2)}\big)^{\top}-\bM_j\bM_j^\top\big\|\\
\leq C_2\|\bA\|_{\rm F}\bigg(\sigma_{\max}(\bM_j)\sqrt{\frac{d_1(t+\log d_{\max})}{n}}+\|\bA\|_{\rm F}\frac{(d_1\ldots d_k)^{1/2}(t+\log d_{\max})}{n}\bigg),
\end{eqnarray*}
which concludes the proof by adjusting the constant $C_2$ and applying Davis-Kahan Theorem.
\end{proof}

\bibliographystyle{plainnat}
\bibliography{refer}

\appendix
\section{Technical Lemmas}
\subsection{Proof of Lemma~\ref{lemma:chernoff}}
Clearly, for any $t$ and $\lambda>0$, 
\begin{eqnarray*}
\PP\Big(\sum_{j=1}^n(X_j-p_j)\geq t\Big)&=&\PP\Big(\exp\Big\{\lambda\sum_{j=1}^n(X_j-p_j)\Big\}\geq \exp\big\{\lambda t\big\}\Big)\\
&\leq& e^{-\lambda t}\EE\exp\Big\{\lambda\sum_{j=1}^n(X_j-p_j)\Big\}\\
&\leq& e^{-\lambda t}\prod_{j=1}^n \EE e^{\lambda(X_j-p_j)}\\
&\leq& e^{-\lambda t}\prod_{j=1}^n \big(p_je^{\lambda(1-p_j)}+(1-p_j)e^{-\lambda p_j}\big).
\end{eqnarray*}
Note that $e^x\leq 1+x+x^2$ for any $x\in[-1,1]$. Then,
$$
p_je^{\lambda(1-p_j)}+(1-p_j)e^{-\lambda p_j}\leq 1+\lambda^2p_j(1-p_j)\leq e^{\lambda^2p_j(1-p_j)}. 
$$
Therefore, we obtain
\begin{eqnarray*}
\PP\Big(\sum_{j=1}^n(X_j-p_j)\geq t\Big)\leq e^{-\lambda t}\prod_{j=1}^n e^{\lambda^2p_j(1-p_j)}=\exp\Big\{-\lambda t+\lambda^2\sum_{j=1}^n p_j(1-p_j)\Big\}.
\end{eqnarray*}
By choosing $\lambda=t/2\sum_{j=1}^np_j(1-p_j)$, we end up with
$$
\PP\Big(\sum_{j=1}^n(X_j-p_j)\geq t\Big)\leq \exp\Big\{-t^2/4\sum_{j=1}^np_j(1-p_j)\Big\}.
$$
The proof is closed after choosing $t=2s\sqrt{\sum_{j=1}^np_j(1-p_j)}$ for $s\geq 0$.

\subsection{Proof of Lemma~\ref{lemma:Dvsqent}}
The proof follows from the same argument as that for Lemma 12 of \cite{yuan2016tensor}. More specifically, denote the aspect ratio for a block $A_1\times\ldots A_k\subset [d_1]\times\ldots\times [d_k]$,
$$
h(A_1\times \ldots\times A_k)=\min\Big\{\nu: |A_j|^2\leq \nu\prod_{j=1}^k|A_j|, j=1,2,\ldots,k\Big\}.
$$
We bound the entropy of a single block. Let
\begin{eqnarray*}
\mathfrak{D}_{\nu, \ell}^{\rm (block)}=\Big\{\sgn(u_1(a_1))\ldots\sgn(u_k(a_k)){\bf 1}\big\{(a_1,\ldots,a_k)\in A_1\times\ldots\times A_k\big\}:\\
h(A_1\times\ldots A_k)\leq \nu, \prod_{j=1}^k|A_j|=\ell\Big\}.
\end{eqnarray*}
By definition, we obtain
$$
\max\big(|A_1|^2,\ldots, |A_k|^2\big)\leq \nu|A_1||A_2|\ldots |A_k|\leq \nu\ell.
$$
By dividing $\mathfrak{D}_{\nu,\ell}^{(\rm block)}$ into subsets according to $(\ell_1,\ldots,\ell_k)=(|A_1|,\ldots,|A_k|)$, we find
$$
\big|\mathfrak{D}_{\nu,\ell}^{\rm (block)}\big|\leq \sum_{\ell_1\ldots\ell_k=\ell, \max_{j}\ell_j\leq \sqrt{\nu\ell}} 2^{\ell_1+\ldots+\ell_k}{d_1\choose \ell_1}\ldots{d_k\choose \ell_k}.
$$
By the Stirling formula, for $j=1,2,\ldots,k$,
$$
{d_j \choose \ell_j}\leq \frac{d_j^{\ell_j}}{(\ell_j!)}\leq \Big(\frac{d_j}{\ell_j}\Big)^{\ell_j}e^{\ell_j}\frac{1}{\sqrt{2\pi \ell_j}},
$$
then
$$
\log\Big[\sqrt{2\pi\ell_j}2^{\ell_j}{d_j\choose \ell_j}\Big]\leq \ell_j L(\ell_j, 2d_{\max})\leq \sqrt{\nu\ell}L(\sqrt{\nu\ell},2d_{\max})
$$
where $L(x,y):=\max\{1,\log(ey/x)\}$. Let $\ell=\prod_{j=1}^{m}p_j^{v_j}$ with distinct prime factors $p_j$. Since $(v_j+1)v_j/(2p_j^{v_j/2})$ is upper bounded by $2.66$ for $p_j=2$, by $1.16$ for $p_j=3$ and by $1$ for $p_j\geq 5$,
we get
\begin{eqnarray*}
\big|\big\{(\ell_1,\ldots,\ell_k): \ell_1\ldots\ell_k=\ell\big\}\big|&=&\prod_{j=1}^m{v_j+1\choose k-1}\\
&\leq& \prod_{j=1}^{m}{v_j+1\choose 2}^{k/2}\\
&\leq& (2.66\times 1.16)^{k/2}(\sqrt{\ell})^{k/2}\\
&\leq& \prod_{j=1}^k\big(2\sqrt{2\pi\ell_j}\big)^{k/2},\qquad \forall \prod_{j=1}^k\ell_j=\ell.
\end{eqnarray*}
Therefore,
\begin{eqnarray*}
\big|\mathfrak{D}_{\nu,\ell}^{\rm(block)}\big|&\leq& \frac{\exp\Big(k\sqrt{\nu \ell}L(\sqrt{\nu\ell}, 2d_{\max})\Big)}{\prod_{j=1}^k\sqrt{2\pi \ell_j}}\prod_{j=1}^k\big(2\sqrt{2\pi \ell_j}\big)^{k/2},\quad \forall (\ell_1\ldots\ell_k)=\ell\\
&\leq& 2^{k^2/2}(2\pi)^{k(k-2)/4}\ell^{(k-2)/4}\exp\Big(k\sqrt{\nu\ell} L(\sqrt{\nu \ell}, 2d_{\max})\Big)\\
&\leq& 2^{k^2/2}(2\pi)^{k(k-2)/4}\exp\Big(2k\sqrt{\nu\ell}L(\sqrt{\nu\ell}, 2d_{\max})\Big).
\end{eqnarray*}

Due to the constraint $b_1+b_2+\ldots+b_k=s$ in defining $\mathfrak{B}^{\star}_{\nu,m_{\star}}$, for any $\bY\in\mathfrak{B}^{\star}_{\nu,m_{\star}}$, $\bD_{s}(\bY)$ is composed of at most $i^{\star}:={s+k-1\choose k-1}$ blocks. Since the sum of the sizes of the blocks is bounded by $2^{q}$, we obtain
\begin{eqnarray*}
\big|\mathfrak{D}_{\nu,s,q}\big|&\leq& \sum_{\ell_1+\ldots+\ell_{i^{\star}}\leq 2^{q}} \prod_{i=1}^{i^{\star}}\big|\mathfrak{D}_{\nu,\ell_i}^{\rm(block)}\big|\\
&\leq& \sum_{\ell_1+\ldots+\ell_{i^{\star}}\leq 2^{q}} (2\pi)^{i^{\star}k(k-2)/4}2^{i^{\star}k^2/2}\exp\Big(2k\sum_{i=1}^{i^{\star}}\sqrt{\nu\ell_i}L(\sqrt{\nu\ell_i}, 2d_{\max})\Big)\\
&\leq& 2^{i^{\star}k^2/2}(2^q)^{i^{\star}}(2\pi)^{i^{\star}k(k-2)/4}\max_{\ell_1+\ldots+\ell_{i^{\star}}\leq 2^q}\exp\Big(2k\sum_{i=1}^{i^{\star}}\sqrt{\nu\ell_i}L(\sqrt{\nu\ell_i}, 2d_{\max})\Big).
\end{eqnarray*}
As shown in \cite{yuan2016tensor}, $\sum_{i=1}^{i^{\star}}\sqrt{\ell_i}L(\sqrt{\nu\ell_i},2d_{\max})\leq \sqrt{i^{\star}2^q}\big(L(\sqrt{\nu2^q},2d_{\max})+\log(\sqrt{i^{\star}})\big)$, we obtain
\begin{eqnarray*}
\log \big|\mathfrak{D}_{\nu,s,q}\big|\leq i^{\star}\log(2^q)+i^{\star}k(k-2)/2+i^{\star}k^2/2+2k\sqrt{i^{\star}\nu2^q}L\big(\sqrt{\nu 2^q}, 2d_{\max}\sqrt{i^{\star}}\big).
\end{eqnarray*}
Since $i^{\star}={s+k-1\choose k-1}\leq s^k$, it follows that 
\begin{eqnarray*}
\log \big|\mathfrak{D}_{\nu,s,q}\big|\leq qs^k\log2+2k^2s^{k}\sqrt{\nu 2^q} L\big(\sqrt{\nu 2^q}, d_{\max}s^{k/2}\big).
\end{eqnarray*}

\end{document}